\documentclass[a4paper,UKenglish,cleveref,autoref]{lipics-v2021}
\hideLIPIcs

\usepackage{amsmath,amssymb,mathtools}
\usepackage{booktabs,tabularx,array}
\usepackage{algorithm}
\usepackage[noend]{algpseudocode}
\algrenewcommand{\algorithmicrequire}{\textbf{Input:}}
\algrenewcommand{\algorithmicensure}{\textbf{Output:}}
\usepackage{tikz,placeins,needspace}
\usetikzlibrary{positioning,calc,arrows.meta}
\newcommand{\OPT}{\operatorname{OPT}}
\newcommand{\eps}{\varepsilon}
\newcommand{\poly}{\operatorname{poly}}
\newcommand{\polylog}{\operatorname{polylog}}
\newcommand{\MP}{\mathcal M}
\newcommand{\rk}{\operatorname{rk}}

\newcommand{\ind}{\operatorname{ind}}
\newcommand{\cstar}{c_\star}
\definecolor{mone}{RGB}{184,52,45}
\definecolor{mtwo}{RGB}{36,91,165}
\definecolor{extra}{RGB}{105,105,105}
\title{Sublinear-Time and Semi-Streaming Algorithms for Path Cover and TSP}
\titlerunning{Sublinear-Time and Semi-Streaming Path Cover and TSP}
\author{Sharareh Alipour}{Tehran Institute for Advanced Studies}{}{}{}
\author{Ali Nowzari}{Sharif university }{}{}{}
\authorrunning{S. Alipour and A. Nowzari}
\Copyright{Sharareh Alipour and Ali Nowzari}
\nolinenumbers

\begin{document}
\raggedbottom
\maketitle
\begin{abstract}
We give improved sublinear-time and semi-streaming approximation algorithms for maximum path cover, $(1,2)$-TSP, and Max-TSP. For each fixed $\eps>0$, our distance-query algorithm estimates the optimum cost of an $n$-vertex $(1,2)$-TSP instance from above within factor $2-\cstar+\eps=1.413184\ldots+\eps$, with high probability, in $n^{2-\Omega_\eps(1)}$ time, where $\cstar=2(\sec1-\tan1)$. This improves the $3/2+\eps$ factor of direct matching estimation in this subquadratic regime~\cite{BhattacharyaKissSaranurak2023} and of the faster near-linear estimator~\cite{BehnezhadEtAl2024}. A single random vertex order defines an implicit bipartite graph whose matchings project to path covers. An alternating-permutation analysis shows that its expected maximum matching size is at least $\cstar$ times the optimum path-cover size. On long paths, even the best of polynomially many uniform orders retains only a $\cstar+o(1)$ fraction with high probability.

In semi-streaming, we construct deterministic $(5/7-\eps)$-approximate unweighted path covers and $(9/7+\eps+5/(7n))$-approximate $(1,2)$-TSP tours, improving the respective $2/3-\eps$ and $4/3+\eps+1/n$ guarantees in~\cite{AlipourFarokhnejadMoemke2025}. For nonnegative weights, we obtain $2/3-\eps$ for path cover and Max-TSP, improving the earlier asymptotic $7/12-\eps$ Max-TSP bound. These algorithms use $\poly(1/\eps)$ passes independent of $n$ and $O(n\log(n/\eps)\poly(1/\eps))$ words. We compute two approximate matchings separated by edge deletion. A cycle-compensation lemma accounts for the weight removed when breaking their cycles. For unweighted graphs, a third approximate matching on a simple auxiliary graph with at most $2n$ vertices chooses component connections; a Tutte--Berge count proves the stronger bound.

We also answer the repeated-contraction question in~\cite{AlipourFarokhnejadMoemke2025} for exact calls with arbitrary tie-breaking: further rounds leave the worst-case path-cover guarantees at $2/3$ unweighted and $5/8$ weighted. Three successive exact residual matchings likewise cannot guarantee more than $2/3$, even with optimal recombination of their edges. Our structured third call avoids this limitation.
\end{abstract}
\clearpage
\section{Introduction}
\label{sec:introduction}
Let $G=(V,E)$ be a simple undirected graph on $n$ vertices. A \emph{path cover}, also called a \emph{linear forest}, is a collection of vertex-disjoint paths covering all vertices, with isolated vertices allowed. The \emph{maximum path cover} (MPC) problem maximizes the number of selected edges; write $\rho(G)$ for its optimum. In \emph{weighted MPC}, each edge has a nonnegative weight, and the objective is to maximize the total selected weight, with optimum $\rho_w(G)$. The problem is NP-hard even in the unweighted case~\cite{AlipourFarokhnejadMoemke2025}.

Path covers also link MPC to the traveling salesman problem (TSP). In $(1,2)$-TSP, one minimizes Hamiltonian-cycle cost in a complete graph with edge costs one or two. Let $G_1=(V,E_1)$ be the graph formed by the cost-one edges, and let $\OPT_{12}$ be the minimum tour cost of the $(1,2)$-TSP instance. A cover with $p$ edges of $G_1$ gives a tour of cost at most $2n-p$. Thus a large cover gives a short tour, while a lower estimate of its size gives an upper estimate of tour cost~\cite{BehnezhadEtAl2024,AlipourFarokhnejadMoemke2025}. In Max-TSP, one maximizes tour weight in a complete nonnegatively weighted graph; completing a weighted cover cannot decrease its weight.

In this paper, we study approximation algorithms for these problems in two settings: \emph{the semi-streaming model} and \emph{sublinear-time computation}, which restrict memory and running time, respectively.

\noindent\textbf{The semi-streaming model.}
Edges are scanned sequentially in a fixed arbitrary order, possibly in several passes, using $O(n\polylog n)$ bits of memory~\cite{FeigenbaumEtAl2005}. Here $\polylog n$ is a fixed polynomial in $\log n$. Our algorithms receive edges with their costs or weights and return explicit path covers or Hamiltonian tours. For fixed accuracy $\eps$, they use $O(n\log n)$ words. If each edge weight is encoded with $O(\polylog n)$ bits, each word also uses $O(\polylog n)$ bits, giving $O(n\polylog n)$ bits overall.

\noindent\textbf{Sublinear-time computation.}
A sublinear-time algorithm accesses selected input entries through queries and uses time smaller than the input size, counting both queries and computation. Here adjacency queries test edges of an unweighted graph, and distance queries return costs in a $(1,2)$-TSP instance. We seek time sublinear in the $\Theta(n^2)$ matrix entries, not necessarily in $n$. The outputs are lower estimates of MPC size or upper estimates of optimum tour cost, not explicit solutions.

\subsection{Notation and path-cover conversion}
\label{sec:preliminaries}
We assume $n\ge3$; MPC instances on one or two vertices are handled directly. Write $m=|E|$ and $uv$ for an edge joining $u$ and $v$. We identify $Q\subseteq E$ with its spanning subgraph $(V,Q)$ and write $\deg_Q(v)$ for the degree of $v$. For nonnegative edge weights $w:E\to\mathbb R_{\ge0}$, put $w(Q)=\sum_{e\in Q}w(e)$.

A matching is a set of pairwise vertex-disjoint edges; $\nu(G)$ and $\nu_w(G)$ denote its maximum size and weight. Alternating along each path partitions a path cover into two matchings. For $S\subseteq E$, let $G-S=(V,E\setminus S)$; for $U\subseteq V$, let $G-U$ be the induced graph obtained by deleting the vertices in $U$.

We write $\OPT_{\mathrm{MaxTSP}}$ for the maximum tour weight in Max-TSP. A \emph{path-preserving completion} is a Hamiltonian cycle containing every edge of a specified path cover.

\noindent\textbf{Accuracy and resources.}
A maximization ratio $\alpha$ guarantees value at least $\alpha$ times optimum; a minimization factor $\lambda$ guarantees cost at most $\lambda$ times optimum. An error-$\delta$ matching call has ratio $1-\delta$, where $0\le\delta<1$; an exact call has $\delta=0$. Ties may be broken arbitrarily. Ratio infima exclude inputs with zero optimum.

In $n^{2-\Omega_\eps(1)}$, the positive exponent gap and hidden constants may depend on fixed $\eps>0$; no polynomial dependence on $1/\eps$ is required. \emph{With high probability} means probability at least $1-n^{-a_0}$ for a fixed confidence exponent $a_0>0$. Logarithms are natural. We use $\poly$ for an unspecified polynomial and $\widetilde O$ to suppress polylogarithmic factors in $n$; a subscript allows dependence on the indicated fixed parameters.

\begin{lemma}
\label{lem:tsp-conversion}
For the cost-one graph $G_1$ of a $(1,2)$-metric, the standard conversion inequalities~\cite{BehnezhadEtAl2024} are
\begin{equation}
 2n-\rho(G_1)-1\le\OPT_{12}\le2n-\rho(G_1),\qquad \OPT_{12}\ge n.
 \label{eq:tsp-conversion}
\end{equation}
If $0\le s\le\rho(G_1)$ and $s\ge\gamma\rho(G_1)-\eta n$, with $0\le\gamma\le1$ and $\eta\ge0$, then $\widehat T=2n-s$ satisfies
\begin{equation}
 \OPT_{12}\le\widehat T\le(2-\gamma+\eta)\OPT_{12}+\gamma.
 \label{eq:common-transfer}
\end{equation}
When $s$ is the size of an explicit cover, completing it gives a tour of cost at most $\widehat T$.
\end{lemma}
\begin{proof}
For~\eqref{eq:common-transfer}, $s\le\rho(G_1)$ and~\eqref{eq:tsp-conversion} give $\widehat T\ge2n-\rho(G_1)\ge\OPT_{12}$. Since $\rho(G_1)\ge2n-\OPT_{12}-1$ and $n\le\OPT_{12}$,
\[
 \widehat T\le\gamma\OPT_{12}+(2-2\gamma+\eta)n+\gamma
 \le(2-\gamma+\eta)\OPT_{12}+\gamma.
\]
If $s$ is the size of an explicit cover, joining its $n-s$ paths costs at most $s+2(n-s)=\widehat T$.
\end{proof}

\subsection{Our results}
\noindent\textbf{Sublinear-time value estimation.}
For every fixed $\eps>0$, we obtain, with high probability, an upper estimate of $(1,2)$-TSP cost within factor $2-\cstar+\eps=1.413184\ldots+\eps$, where $\cstar=2(\sec1-\tan1)$, in $n^{2-\Omega_\eps(1)}$ time and queries. Table~\ref{tab:main-results} compares our approximation factor and running time with the known bounds. Our factor is smaller than $3/2+\eps$, while our time bound is larger than that of the near-linear estimator.

Our estimator uses a single uniform vertex order to orient the edges, then splits each vertex into incoming and outgoing copies. Matchings in the resulting implicit bipartite graph project to path covers; each query requires at most one original query. Since $\nu(G)\le\rho(G)\le2\nu(G)$, direct additive matching estimation retains coefficient $1/2$. Our representation raises the expected coefficient to $\cstar=0.586815986\ldots$. An alternating-run analysis using Andr\'e's identity~\cite{Stanley2009} proves this bound. Concentration makes one order sufficient for large $n$; Algorithm~\ref{alg:query-cover} combines it with additive matching-value estimation~\cite{BhattacharyaKissSaranurak2023}. On long paths, even the best of polynomially many uniform vertex orders gives a maximum matching of size at most $(\cstar+o(1))$ times the path length, with high probability.

\begin{table}[htbp]
\centering\small
\caption{Approximation factors and running times for sublinear-time $(1,2)$-TSP cost estimation, for fixed accuracy.}
\label{tab:main-results}
\begin{tabularx}{\textwidth}{@{}>{\raggedright\arraybackslash}Xll@{}}
\toprule
Method and source & Factor & Time / pair-query bound\\\midrule
Estimator~\cite{ChenKannanKhanna2020} & $13/8+\eps$ & $\widetilde O_\eps(n^{3/2})$\\
Estimator~\cite{BehnezhadEtAl2024} & $3/2+\eps$ & $\widetilde O_\eps(n)$\\
Direct matching estimate~\cite{BhattacharyaKissSaranurak2023} & $3/2+\eps$ & $n^{2-\Omega_\eps(1)}$\\
This paper & $1.413184\ldots+\eps$ & $n^{2-\Omega_\eps(1)}$\\\bottomrule
\end{tabularx}
\end{table}

\noindent\textbf{Semi-streaming constructions.}
For $0<\eps<1/3$, we obtain deterministic $(5/7-\eps)$-approximations for unweighted MPC and $(9/7+\eps+5/(7n))$-approximations for $(1,2)$-TSP. For nonnegative weights, we obtain $(2/3-\eps)$-approximations for weighted MPC and Max-TSP. All these algorithms use $\poly(1/\eps)$ passes independent of $n$ and $O(n\log(n/\eps)\poly(1/\eps))$ words. Table~\ref{tab:streaming-comparison} shows the improved guarantees within the same pass regime~\cite{AlipourFarokhnejadMoemke2025}.

Our idea is to delete the edges of a first matching rather than contract them as in~\cite{AlipourFarokhnejadMoemke2025}. Algorithm~\ref{alg:deletion} computes a matching $A$ in $G$ and a matching $B$ in $G-A$, keeping all vertices. A \emph{cycle-compensation lemma} relates the weight lost when opening cycles of $A\cup B$ to the residual matching optimum, yielding the weighted bound for approximate calls. For unweighted MPC, Algorithm~\ref{alg:optimized-joining} chooses component connections with a third matching before opening cycles. A Tutte--Berge count bounds the missed path-cover edges. A simple auxiliary graph on at most $2n$ vertices implements this optimization without storing all component adjacencies, and its matching guarantee transfers to the entire output forest.

\begin{table}[htbp]
\centering\small
\caption{Semi-streaming constructions in $\poly(1/\eps)$ passes independent of $n$.}
\label{tab:streaming-comparison}
\begin{tabularx}{\textwidth}{@{}>{\raggedright\arraybackslash}X>{\raggedright\arraybackslash}Xl@{}}
\toprule
Problem & Comparison and source & This paper\\\midrule
Unweighted MPC & $2/3-\eps$~\cite{AlipourFarokhnejadMoemke2025} & $5/7-\eps$\\
$(1,2)$-TSP & $4/3+\eps+1/n$~\cite{AlipourFarokhnejadMoemke2025} & $9/7+\eps+5/(7n)$\\
Max-TSP & $(7/12-3/(4n))(1-\eps)$~\cite{AlipourFarokhnejadMoemke2025} & $2/3-\eps$\\
Weighted MPC & $1/2-\eps$ via matching~\cite{HuangSu2023} & $2/3-\eps$\\\bottomrule
\end{tabularx}
\end{table}

We also answer the iteration question in~\cite{AlipourFarokhnejadMoemke2025} for exact matching calls. Three unstructured residual matchings have retained-edge path-cover infimum $2/3$; our structured third call avoids this obstruction. Appendix~\ref{app:component-extensions} gives stronger unweighted guarantees with $\log n$-dependent passes and subquadratic postprocessing when a starting $2$-matching is supplied explicitly.

\noindent\textbf{Further related work.}
Matching pairs and successive selection are classical~\cite{ChenKannanKhanna2020,FeigeOfekWieder2002}. The offline weighted factor $2/3$ is also known through a degree-two relaxation followed by cycle breaking~\cite{MoranNewmanWolfstahl1990}. Offline unweighted guarantees include $3/4$ through triangle-free $2$-matching and $6/7$~\cite{Paluch2025,BermanKarpinski2006}, but these do not imply our access bounds. TSP under both streaming and distance-query access is also studied in~\cite{ChenKhannaTan2023}.

The time budget is essential to the sublinear comparison. Beating $3/2$ in $n^{1+o(1)}$ time would settle a near-linear bipartite-matching promise problem~\cite{BehnezhadEtAl2024}; our larger time budget does not resolve it. Some fixed $(1+\eps_0)$ TSP approximation requires $\Omega(n^2)$ queries~\cite{ChenKannanKhanna2020}. For additive matching estimation, every exponent gap $\zeta>0$ admits an accuracy requiring $\Omega(n^{2-\zeta})$ queries in the adjacency-list~\cite{BehnezhadRoghaniRubinstein2024} and pair-query~\cite{AzarmehrEtAl2025} models. 

\noindent\textbf{Organization of the paper.}
Section~\ref{sec:query} gives sublinear-time path-cover and TSP estimates. Sections~\ref{sec:deletion} and~\ref{sec:optimized-joining} analyze the weighted and unweighted constructions; Section~\ref{sec:streaming} gives their semi-streaming implementations. Section~\ref{sec:limits-summary} presents algorithmic limits and open questions. Appendix~\ref{sec:limits} proves the residual-matching obstruction, and Appendix~\ref{sec:contraction} gives a short proof of the unweighted two-round guarantee and proves the contraction limits. Appendix~\ref{app:comparisons} collects the conditional weighted comparison, compatible attachments, the one-call alternative, and hardness. Appendix~\ref{app:component-extensions} gives the general component bound, stronger input-size-dependent pass guarantees, and estimation with a supplied $2$-matching. Appendix~\ref{sec:greedy-joining} gives the lower-pass greedy algorithm and its upper family. Appendix~\ref{app:bks-parameters} derives an explicit subquadratic exponent.

\section{Sublinear-time path-cover and TSP estimation}
\label{sec:query}
We estimate path-cover size from adjacency-matrix (pair) queries, using a bipartite graph determined by a randomly sampled vertex order.

\Needspace{11\baselineskip}
\subsection{Random-order bipartite representation}
\label{sec:split-graph}
Choose a uniformly random ordering of $V$, and let $\pi(v)$ be the position of $v$. Let $V_L=\{v_L:v\in V\}$ and $V_R=\{v_R:v\in V\}$ be disjoint copies of $V$. In the bipartite graph $B_\pi$, represent each original edge by an edge joining the left copy of its earlier endpoint to the right copy of its later endpoint: $u_Lv_R\in E(B_\pi)\quad\Longleftrightarrow\quad uv\in E(G)\text{ and }\pi(u)<\pi(v)$.
A uniform shuffle generates the stored positions in $O(n)$ word operations. We answer queries without storing the edges of $B_\pi$. For a cross-side query, write its endpoints as $u_L,v_R$. Return no if $\pi(u)\ge\pi(v)$; otherwise query whether $uv\in E(G)$. Same-side queries return no. Thus each auxiliary query uses at most one original query, regardless of the order of its endpoints.

\begin{lemma}
\label{lem:projection}
Replacing each edge $u_Lv_R$ of a matching in $B_\pi$ by $uv$ gives a path cover of $G$ with the same number of edges. Thus $0\le\nu(B_\pi)\le\rho(G)$ for every order $\pi$.
\end{lemma}
\begin{proof}
Each vertex $v$ has two copies, and a matching uses at most one edge at each copy. The resulting subgraph of $G$ therefore has degree at most two. Suppose it contains a cycle, and choose the vertex $v$ with smallest $\pi(v)$ on that cycle. Both cycle neighbors occur later than $v$, so the two cycle edges would both use $v_L$ in $B_\pi$, contradicting the matching property. Thus the subgraph is a disjoint union of paths. Including unused vertices as isolates gives a path cover. Each original edge has only one representing edge, so the edge count is unchanged.
\end{proof}

\subsection{Expected path-cover retention}
Fix a path $P=v_0\cdots v_\ell$ with edges $e_i=v_{i-1}v_i$ for $1\le i\le\ell$. Two path edges \emph{conflict} if their representing edges in $B_\pi$ share a vertex and hence cannot both belong to a matching. Nonconsecutive edges have disjoint endpoints and cannot conflict. Thus $e_i$ and $e_{i+1}$ conflict precisely when $\pi(v_i)$ is smaller than both neighboring positions or larger than both: their representing edges then share $(v_i)_L$ or $(v_i)_R$, respectively. These are the local minima and maxima of the positions along $P$. When $\pi(v_i)$ lies between its neighbors' positions, the two edges use different copies and are compatible.

The \emph{conflict graph} $C_\pi$ has one vertex for each $e_i$, with two vertices adjacent exactly when the corresponding path edges conflict. Its components are paths or isolates, since only consecutive $e_i$ can be adjacent. An independent set contains no adjacent pair; write $\ind(C)$ for its maximum size in $C$. Thus $\nu(B_\pi[P])=\ind(C_\pi)$, where $B_\pi[P]$ keeps only edges representing those of $P$. Figure~\ref{fig:orientation} illustrates the correspondence.

\begin{figure}[htbp]
\centering
\begin{tikzpicture}[x=1.0cm,y=1cm,>=Stealth]
\foreach \i/\r in {0/1,1/4,2/2,3/3,4/5}{
 \node[circle,draw,inner sep=1.9pt] (p\i) at (\i,0) {$v_{\i}$};
 \node[above=5pt] at (p\i.north) {$\r$};}
\draw[thick] (p0)--node[below]{$e_1$}(p1);
\draw (p2)--node[below]{$e_2$}(p1);
\draw[thick] (p2)--node[below]{$e_3$}(p3);
\draw[thick] (p3)--node[below]{$e_4$}(p4);
\node at (2,-.9) {(a) Path with vertex positions};
\begin{scope}[shift={(7,0)}]
\foreach \i in {1,2,3,4}{\node[circle,draw,inner sep=2.4pt] (c\i) at ({.9*(\i-1)},0) {$e_{\i}$};}
\draw (c1)--(c2)--(c3);
\node at (1.35,-.9) {(b) Conflict graph};
\end{scope}
\end{tikzpicture}
\caption{Positions $1,4,2,3,5$ make $e_1,e_2$ share $(v_1)_R$ and $e_2,e_3$ share $(v_2)_L$ in $B_\pi$. The compatible set $\{e_1,e_3,e_4\}$, shown bold in (a), gives a matching in $B_\pi$ and two paths in $G$.}
\label{fig:orientation}
\end{figure}
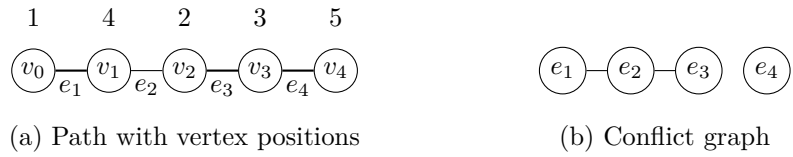

A component of $C_\pi$ with $r$ vertices contributes $\lceil r/2\rceil$ compatible edges. We bound their expected total using the constant $\cstar=2(\sec1-\tan1)$ defined in the introduction.

\begin{lemma}
\label{lem:sharp-retention}
For every path with $\ell\ge1$ edges, $\cstar\ell\le\mathbb E[\ind(C_\pi)]\le\cstar\ell+2$.
The coefficient $\cstar$ in the lower bound is best possible.
\end{lemma}
\begin{proof}
\emph{Proof idea.} Conflicts along a path form a stationary process of alternating runs; the density of an optimal independent set is an alternating sum of run probabilities, which Andr\'e's identity evaluates in closed form.

\emph{Conflict runs.} Let $(U_i)_{i\in\mathbb Z}$ be independent uniform ranks in $(0,1)$ assigned to vertices $i$. Conflict vertex $i$ represents the path edge between vertices $i$ and $i+1$; join conflict vertices $i-1,i$ when $U_i$ is a local extremum. The ranks have the same independent distribution at every position, and the conflict rule is unchanged by shifting all indices. Hence probabilities of conflict patterns do not depend on their starting position; this is stationarity.

Let $p_k$ be the probability that $k$ consecutive edges of the conflict graph are present, with $p_0=1$. The $k+2$ ranks involved must alternate, in either direction. If $E_j$ denotes the Euler zigzag numbers (the number of up--down permutations of $j$ elements, with $E_0=E_1=1$), then $p_k=\frac{2E_{k+2}}{(k+2)!}\quad(k\ge0)$.
For $k\ge1$, choose the requirements at centers $1,4,7,\ldots\le k$, a total of $\lceil k/3\rceil$. They use disjoint triples of independent ranks, so
\begin{equation}
 p_k\le(2/3)^{\lceil k/3\rceil}\qquad(k\ge0).
 \label{eq:run-decay}
\end{equation}
By \eqref{eq:run-decay}, an infinite run from any fixed index has probability zero. A countable union over all indices shows that every component is finite almost surely.

\emph{Independent-set density.} Let $q_r$ be the probability that a fixed conflict vertex starts a component of exactly $r$ vertices. Inclusion--exclusion at its two boundaries gives
\begin{equation}
 q_r=p_{r-1}-2p_r+p_{r+1},\qquad \sum_{r\ge1}r q_r=1.
 \label{eq:component-probabilities}
\end{equation}
The second identity counts the $r$ possible positions of a vertex in an $r$-vertex component. Select odd-positioned vertices in each component, starting from its left endpoint. This gives an independent set that is optimal on every component, with the same selection probability at each index:
\begin{align}
 d&=\sum_{r\ge1}\lceil r/2\rceil q_r
   =\frac12+\frac12\sum_{j\ge0}q_{2j+1}
   =\sum_{k\ge0}(-1)^kp_k.\label{eq:density}
\end{align}
For the last equality, substitute \eqref{eq:component-probabilities}: the odd-order sum is $p_0+2\sum_{k\ge1}(-1)^kp_k$. Bound~\eqref{eq:run-decay} gives absolute convergence and justifies the rearrangements.

\emph{Evaluating the density.} Andr\'e's identity, as stated in Stanley's survey~\cite{Stanley2009}, is $\sum_{j\ge0}E_j\frac{z^j}{j!}=\sec z+\tan z$.
Its radius of convergence is $\pi/2$, so it converges at $z=-1$. Substitution in \eqref{eq:density}, with $E_0=E_1=1$, gives $d=2\sum_{j\ge2}(-1)^j\frac{E_j}{j!} =2(\sec(-1)+\tan(-1)-1+1)=\cstar$.

\emph{Finite paths.} Restrict the stationary independent set to $\ell$ consecutive conflict vertices. It remains independent and has expected size $\cstar\ell$. The $\ell+1$ underlying ranks have the same relative-order distribution as on the finite path, proving the lower bound.

For the upper bound, the finite components of orders $r_1,\ldots,r_j$ have total optimum $\sum_i\lceil r_i/2\rceil$. The stationary choice is optimal on every untruncated component. On a truncated component it selects either $\lfloor r_i/2\rfloor$ or $\lceil r_i/2\rceil$ vertices, losing at most one. Only the two boundary components can be truncated. The finite optimum therefore exceeds the restricted choice by at most two. Dividing by $\ell$ proves sharpness.
\end{proof}

\begin{theorem}
\label{thm:orientation-retention}
For every finite simple graph $G$,
\[
 \cstar\rho(G)\le\mathbb E_\pi[\nu(B_\pi)]\le\rho(G).
\]
The coefficient $\cstar$ cannot be increased uniformly over all graphs.
\end{theorem}
\begin{proof}
Fix an optimum path cover independently of $\pi$. Its path components have uniform restricted orders and vertex-disjoint split graphs. Unite maximum matchings of these graphs and apply \Cref{lem:sharp-retention} by linearity of expectation. This proves the lower bound; \Cref{lem:projection} gives the upper bound for every order. Long paths establish sharpness.
\end{proof}

\subsection{Concentration and sharpness}
\begin{theorem}
\label{thm:order-concentration}
Let $G$ be a graph on $n\ge2$ vertices. For any integer $r\ge1$ and random orders $\pi_1,\ldots,\pi_r$, each marginally uniform, put $X_j=\nu(B_{\pi_j})$. For $u>0$,
\[
 \Pr\!\left[\max_j\left|X_j-\mathbb E X_j\right|>u\right]
 \le 2r\exp\!\left(-\frac{u^2}{2n}\right).
\]
If $G$ is a path and $r\le n^D$ for a fixed $D$, then $\max_jX_j=\cstar(n-1)+O(\sqrt{n\log n})$ with high probability. Independence between the orders is not required.
\end{theorem}
\begin{proof}
For each fixed $j$, generate its uniform marginal from independent continuous ranks. Changing one rank moves one vertex $v$ and leaves the relative order of all other vertices unchanged, so $B_\pi-\{v_L,v_R\}$ is the same graph before and after. Deleting a vertex lowers $\nu$ by at most one, so $\nu(B_\pi)$ changes by at most two. Bounded differences~\cite{McDiarmid1989} gives $\Pr[|X_j-\mathbb E X_j|>u]\le2e^{-u^2/(2n)}$. A union bound proves the first claim. For a path, take $u=\sqrt{2n\log(2r n^{a_0})}$, where $a_0>0$, and use $\cstar(n-1)\le\mathbb E X_j\le\cstar(n-1)+2$ from \Cref{lem:sharp-retention}.
\end{proof}
On the $(1,2)$-metric whose cost-one graph is a path on $n$ vertices, $\OPT_{12}=n+1$ for $n\ge3$. With polynomially many uniform orders, the ratio of $2n-\max_jX_j$ therefore converges to $2-\cstar$ with high probability. Lower estimates $z_j\le X_j$ cannot improve it: subtracting $\max_jz_j$ from $2n$ only increases the reported cost. If also $z_j\ge X_j-\eta n$ for every $j$, that increase is at most $\eta n$, or less than $\eta$ in the ratio. This error need not vanish for fixed $\eta$, so exact convergence is asserted for exact matching values.

\subsection{Path-cover and TSP value estimation}
We use an additive estimator that returns a lower bound on matching size.
\begin{theorem}[\cite{BhattacharyaKissSaranurak2023}]
\label{thm:bks-interface}
For every fixed $0<a<1$, an $N$-vertex graph $H$ given by adjacency-matrix queries has a randomized estimator returning, with high probability, a value $z$ satisfying $\nu(H)-aN\le z\le\nu(H)$. Its running time is $T_{\rm add}(N,a)=\widetilde O_a(N^{2-\xi(a)})$ for a positive $\xi(a)$.
\end{theorem}
The reduction uses only this value-estimation interface. Independent median repetitions raise its confidence while preserving the interval; Appendix~\ref{app:bks-parameters} gives an explicit admissible $\xi$. By \Cref{thm:order-concentration}, one vertex order suffices outside a bounded-size branch.

\begin{algorithm}[H]
\caption{Single-order path-cover value estimation}
\label{alg:query-cover}
\begin{algorithmic}[1]
\Require Adjacency-query access to $G$ on $n\ge2$ vertices; $0<\eta<1$; $a_0>0$.
\Ensure A scalar estimate $s$ of the maximum path-cover size.
\If{$n<8\eta^{-2}\log(4n^{a_0})$}
  \State Query all vertex pairs and \Return $\rho(G)$, computed by exhaustive search.
\EndIf
\State Sample a uniform vertex order $\pi$ and store its positions.
\State Provide the implicit adjacency oracle for $B_\pi$ from Section~\ref{sec:split-graph}.
\State Using fresh randomness, estimate $\nu(B_\pi)$ as $z$ at accuracy $\eta/4$ (\Cref{thm:bks-interface}),
\Statex \hspace{\algorithmicindent}amplified to failure probability at most $n^{-a_0}/2$.
\State \Return $\max\{0,\min\{n-1,z\}\}$.
\end{algorithmic}
\end{algorithm}

\Needspace{6\baselineskip}
\begin{theorem}
\label{thm:query-cover}
Fix $0<\eta<1$ and $a_0>0$. On an $n\ge2$ vertex graph given by adjacency queries, Algorithm~\ref{alg:query-cover} runs in $n^{2-\Omega_\eta(1)}$ time and, with probability at least $1-n^{-a_0}$, returns $s$ satisfying
\begin{equation}
 \cstar\rho(G)-\eta n\le s\le\rho(G).
 \label{eq:query-cover}
\end{equation}
The hidden constants and exponent may depend on $\eta,a_0$. No path cover is returned.
\end{theorem}
\begin{proof}
\emph{Success probability.} The exhaustive branch is exact. Otherwise put $Y=\nu(B_\pi)$. By \Cref{thm:orientation-retention}, $\mathbb EY\ge\cstar\rho(G)$. Taking $r=1$ and $u=\eta n/2$ in \Cref{thm:order-concentration} gives $\Pr[Y<\cstar\rho(G)-\eta n/2]\le2e^{-\eta^2n/8}\le n^{-a_0}/2$, using the threshold in Algorithm~\ref{alg:query-cover}. Conditional on the sampled order, the matching estimate satisfies $Y-\eta n/2\le z\le Y$ except with probability at most $n^{-a_0}/2$. On the joint success event, $\cstar\rho(G)-\eta n\le z\le Y\le\rho(G)$. Clipping preserves these bounds because $0\le\rho(G)\le n-1$. A union bound proves~\eqref{eq:query-cover}.

\emph{Time and queries.} Inputs failing the threshold have bounded size for fixed $\eta,a_0$, so the exhaustive work is absorbed into a parameter-dependent constant. Otherwise one shuffle takes $O(n)$ operations, and each auxiliary query costs at most one original query. The $O_{\eta,a_0}(\log n)$ confidence repetitions give
\begin{equation}
 O_{\eta,a_0}\!\left(\log n\,T_{\rm add}(2n,\eta/4)+n\right)
 =\widetilde O_{\eta,a_0}(n^{2-\xi(\eta/4)}).
 \label{eq:query-composition-time}
\end{equation}
Absorbing the logarithmic factors gives $n^{2-\Omega_\eta(1)}$ time.
\end{proof}
A purely multiplicative guarantee on all graphs would require quadratic queries: a uniformly hidden edge is found by $q$ distinct pair queries with probability at most $q/\binom n2$. Until then, the answers coincide with those of the empty graph. This explains the additive term in~\eqref{eq:query-cover}.

\begin{theorem}
\label{thm:query-tsp}
For every $n$-vertex $(1,2)$-TSP instance with $n\ge3$, every fixed $0<\eps<1$, and every confidence exponent $a_0>0$, a randomized distance-query algorithm returns, with probability at least $1-n^{-a_0}$,
\[
 \OPT_{12}\le\widehat T\le(2-\cstar+\eps)\OPT_{12}
\]
in $n^{2-\Omega_\eps(1)}$ time and queries. More explicitly, the bound is $\widetilde O_{\eps,a_0}(n^{2-\xi(\eps/8)})$, where $\xi$ is specified in Appendix~\ref{app:bks-parameters}.
\end{theorem}
\begin{proof}
For $n\ge\lceil2/\eps\rceil$, run Algorithm~\ref{alg:query-cover} on $G_1$ with $\eta=\eps/2$ and return $\widehat T=2n-s$. By \Cref{thm:query-cover,lem:tsp-conversion}, on the success event,
\[
 \OPT_{12}\le\widehat T\le(2-\cstar+\eps/2)\OPT_{12}+\cstar
                  \le(2-\cstar+\eps)\OPT_{12},
\]
since $\OPT_{12}\ge n\ge2/\eps$ and $\cstar<1$. For smaller $n$, read the whole instance and solve it exhaustively. Its size depends only on fixed $\eps$, so the work is absorbed into the $\eps$-dependent constant. Each query to $G_1$ costs one distance query.
\end{proof}

\begin{remark}
\label{rem:elementary-retention}
A simpler analysis gives a weaker approximation guarantee. For a path with $\ell\ge1$ edges, let $Z$ count the isolated vertices of its conflict graph $C_\pi$. Since each component is a path, $\ind(C_\pi)\ge\ell/2+Z/2$. At an interior position, a conflict vertex is isolated with probability $1/12$, because the four underlying ranks must be monotone. At either endpoint position, this probability is $1/3$ when $\ell\ge2$; the sole vertex is always isolated when $\ell=1$. Thus $\mathbb E[Z]\ge\ell/12$ and $\mathbb E[\ind(C_\pi)]\ge13\ell/24$. Replacing $\cstar$ by $13/24$ in the estimation argument therefore gives the weaker $(35/24+\eps)$-approximation for $(1,2)$-TSP cost estimation, without the alternating-permutation calculation.
\end{remark}
\section{Weighted path cover by edge deletion}
\label{sec:deletion}
Algorithm~\ref{alg:deletion} gives the two-call deletion rule for weighted path cover. Section~\ref{sec:streaming} handles weight normalization and the small-instance Max-TSP branch.

\begin{algorithm}[H]
\caption{Weighted path cover by edge deletion}
\label{alg:deletion}
\begin{algorithmic}[1]
\Require Graph $G=(V,E)$ with nonnegative weights $w$; ratios $0<\alpha,\beta\le1$.
\Ensure A path cover $(V,F)$.
\State Compute a matching $A$ in $G$ with $w(A)\ge\alpha\nu_w(G)$.
\State Compute a matching $B$ in $G-A$ with $w(B)\ge\beta\nu_w(G-A)$.
\State $F\gets A\cup B$.
\For{each cycle component $C$ of $(V,A\cup B)$}
  \State Delete from $F$ an edge of minimum weight in $C$.
\EndFor
\State \Return $(V,F)$.
\end{algorithmic}
\end{algorithm}

Let $\mathcal C$ be the cycle components of $A\cup B$. Their edges alternate between the matchings, so each cycle has even length at least four. Deleting one edge per cycle therefore gives a path cover; the deleted edge may belong to either matching. For each cycle $C$, let $\lambda_C=\min_{e\in E(C)}w(e)$, and put $\Lambda=\Lambda(A,B)=\sum_{C\in\mathcal C}\lambda_C$.
Thus $w(F)=w(A)+w(B)-\Lambda$, with $\Lambda=0$ when there are no cycles. Each cycle has at least two $B$-edges, so
\begin{equation}
 w(B)\ge2\Lambda.\label{eq:cycle-loss}
\end{equation}

\begin{lemma}
\label{lem:compensation}
For any edge-disjoint matchings $A,B$ of a nonnegatively weighted graph $G$ and any linear forest $P\subseteq G$,
\[
 2\nu_w(G-A)\ge w(P)-w(A)+\Lambda(A,B).
\]
\end{lemma}
\begin{proof}
Put $P_0=P-A$. Partition $\mathcal C$ into $\mathcal C_0$, whose cycles have an $A$-edge outside $P$, and $\mathcal C_1$, whose $A$-edges all lie in $P$. A missing $A$-edge from each cycle of $\mathcal C_0$ gives distinct edges of weight at least $\lambda_C$. Nonnegativity therefore gives
\begin{equation}
 w(A\setminus P)\ge\sum_{C\in\mathcal C_0}\lambda_C.
 \label{eq:missing-cycles}
\end{equation}
For the other cycles, we will select one new $B$-edge per cycle whose addition to $P_0$ preserves a linear forest.

\emph{Simultaneous selection.} The graphic matroid of $G$ treats an edge set as independent exactly when it contains no cycle. Its rank $\rk(X)$ is the maximum size of an acyclic subset of $X$; equivalently, $\rk(X)=n-\operatorname{cc}(X)$, where $\operatorname{cc}(X)$ counts the components of $(V,X)$, including isolates. Contract each component of $P_0$ to one vertex and keep the edges of $E\setminus P_0$ as distinct copies, including loops and parallel edges. In the resulting graphic matroid, a set $X\subseteq E\setminus P_0$ has rank $\rk(P_0\cup X)-\rk(P_0)$ and is independent exactly when $P_0\cup X$ is acyclic.

For $\mathcal J\subseteq\mathcal C_1$, define
$A_{\mathcal J}=A\cap\bigcup_{C\in\mathcal J}E(C)$ and
$B_{\mathcal J}=B\cap\bigcup_{C\in\mathcal J}E(C)$.
Since $P_0\cup A_{\mathcal J}\subseteq P$ is a forest, its rank is $\rk(P_0)+|A_{\mathcal J}|$.
In $P_0\cup A_{\mathcal J}\cup B_{\mathcal J}$, delete one $A$-edge per cycle in $\mathcal J$. Each deleted edge has a bypass along its own cycle; the cycles are disjoint, so these deletions preserve connectivity. The remaining graph adds only $|A_{\mathcal J}|-|\mathcal J|$ edges to $P_0\cup B_{\mathcal J}$. Hence $\rk(P_0)+|A_{\mathcal J}| \le\rk(P_0\cup A_{\mathcal J}\cup B_{\mathcal J}) \le\rk(P_0\cup B_{\mathcal J})+|A_{\mathcal J}|-|\mathcal J|$,
and therefore
\begin{equation}
 \rk(P_0\cup B_{\mathcal J})-\rk(P_0)\ge|\mathcal J|.
 \label{eq:rado-condition}
\end{equation}
The disjoint candidate sets are $(B\cap E(C))\setminus P_0$, for $C\in\mathcal C_1$. Rado's criterion~\cite{Rado1942} gives one representative from each set, independent in the contracted matroid, whenever the union of any subfamily has rank at least the number of sets in that subfamily. Equation~\eqref{eq:rado-condition} verifies this condition; excluding edges already in $P_0$ does not change the rank. The selected edges $e_C$ therefore create no cycle when added to $P_0$.

\emph{Degrees and weight.} Each cycle in $\mathcal C_1$ has its $A$-edges in $P$, so deleting $A$ leaves its vertices with degree at most one. The representatives belong to the matching $B$, so their endpoints are distinct. Hence $Q=P_0\cup\{e_C:C\in\mathcal C_1\}$ is a path cover in $G-A$. Splitting $Q$ into two matchings and using~\eqref{eq:missing-cycles} gives $2\nu_w(G-A)\ge w(Q)\ge w(P-A)+\sum_{C\in\mathcal C_1}\lambda_C =w(P)-w(A)+w(A\setminus P)+\sum_{C\in\mathcal C_1}\lambda_C \ge w(P)-w(A)+\Lambda(A,B)$.
\end{proof}

\begin{theorem}
\label{thm:deletion}
Let $0<\alpha,\beta\le1$. Suppose $A$ is a matching in $G$ and $B$ a matching in $G-A$ with $w(A)\ge\alpha\nu_w(G)$ and $w(B)\ge\beta\nu_w(G-A)$. Then Algorithm~\ref{alg:deletion} returns a path cover $F$ satisfying $w(F)\ge\frac{2\alpha+\beta-\alpha\beta}{4-\beta}\rho_w(G)$.
Exact calls give $2/3$. With two error-$\delta$ calls, $0\le\delta<1$, $w(F)\ge\frac{(1-\delta)(2+\delta)}{3+\delta}\rho_w(G) \ge\frac23(1-\delta)\rho_w(G)$.
\end{theorem}
\begin{proof}
Splitting an optimum cover into two matchings gives $w(A)\ge\alpha \rho_w(G)/2$. By \Cref{lem:compensation} and the second-call guarantee, $\beta \rho_w(G)\le\beta w(A)+2w(B)-\beta\Lambda$. Use the identity $(4-\beta)w(F)=(\beta w(A)+2w(B)-\beta\Lambda)+(4-2\beta)w(A)+(2-\beta)(w(B)-2\Lambda)$.
Its terms are bounded below by $\beta \rho_w(G)$, $\alpha(2-\beta)\rho_w(G)$, and zero, respectively, using~\eqref{eq:cycle-loss}. This proves the ratio. Set $\alpha=\beta=1$ for the exact case, or $\alpha=\beta=1-\delta$ and use $(2+\delta)/(3+\delta)\ge2/3$ for the error bound.
\end{proof}

\noindent\textbf{Role of compensation.}
The term $\Lambda(A,B)$ in \Cref{lem:compensation} accounts for the weight later removed by cycle breaking. Omitting it gives only $w(A)\ge\rho_w(G)/2$ and $w(B)\ge(\rho_w(G)-w(A))/2$ for exact calls, hence $w(A)+w(B)\ge3\rho_w(G)/4$~\cite{FeigeOfekWieder2002}. Separately bounding the loss by one quarter of this total guarantees only $9\rho_w(G)/16$~\cite{BehnezhadEtAl2024}. Keeping the loss term in the residual-matching bound gives $2/3$ instead. This exact-call constant is tight (\Cref{prop:deletion-tight}); tightness of the general $(\alpha,\beta)$ bound is not established.

\section{Unweighted path cover by component matching}
\label{sec:optimized-joining}
For unweighted graphs, we choose connections between the components of $A\cup B$ before opening their cycles. A matching in a simple auxiliary graph implements this choice.

\subsection{Component graph and matching bounds}
Let $Q\subseteq E$ be a \emph{simple $2$-matching}: a spanning subgraph of maximum degree at most two, with cycles allowed. A vertex is \emph{available} if it lies on a cycle of $Q$ or has degree at most one in $Q$. An \emph{eligible connector} is an original edge joining available vertices in distinct components. The simple \emph{component graph} $K_Q$ has the components of $(V,Q)$ as vertices and an adjacency whenever an eligible connector exists. Let $k$ count the cycles of $Q$. Deleting one edge from each cycle leaves $f=|Q|-k$ edges in $n-f$ path components, including isolates. Let $h_*=\nu(K_Q)$ be the maximum number of disjoint component pairs that can be joined. A matching of $K_Q$ with $j$ edges gives a cover of size $f+j$: choose a witnessing connector for each matched pair, open each matched cycle at its attachment vertex by deleting an incident edge, and open other cycles arbitrarily. The connectors join disjoint pairs of paths, so no cycle or degree-three vertex is created.

For each component $C$, let $P_C$ be its available-vertex set and put $q_C=|P_C|\ge1$. Let $I$ contain the internal (degree-two) vertices of the path components of $Q$ and put $i=|I|$.

\begin{lemma}
\label{lem:component-deficiency}
There exists $S\subseteq V(K_Q)$ with $|S|\le h_*$ such that
\begin{equation}
 \rho(G)\le f+2h_*+i+\sum_{C\in S}(q_C-1).
 \label{eq:component-deficiency}
\end{equation}
\end{lemma}
\begin{proof}
Let $o(H)$ count the odd-order connected components of a graph $H$. Since $K_Q$ has $n-f$ vertices, the Tutte--Berge formula~\cite{QuWest2023} gives $S$ with $o(K_Q-S)-|S|=n-f-2h_*$. The bound $o(K_Q-S)\le n-f-|S|$ then gives $|S|\le h_*$.

An optimum path cover $P$ has $n-\rho(G)$ components. Delete the vertices in $I$ and all sets $P_C$ with $C\in S$. Each deleted vertex increases the number of components by at most one, so at most $n-\rho(G)+i+\sum_{C\in S}q_C$ remain.

Every component of $K_Q-S$ contains a nonempty surviving set $P_C$. No surviving input edge joins sets in different components of $K_Q-S$, because such an edge would be eligible. Thus at least $o(K_Q-S)$ forest components remain, and $n-f-2h_*+|S|\le n-\rho(G)+i+\sum_{C\in S}q_C$.
Rearranging proves the bound. The algorithm need not compute $S$.
\end{proof}

\begin{lemma}
\label{lem:optimized-blocking}
If $Q=A\cup B$ for edge-disjoint matchings, set $s=|B|-2k\ge0$. Then
\begin{equation}
 \rho(G)\le f+2s+5h_*.
 \label{eq:optimized-blocking}
\end{equation}
\end{lemma}
\begin{proof}
Use $S$ from \Cref{lem:component-deficiency} and put $b_C=|B\cap E(C)|$. Every cycle has $2b_C\ge4$ edges. A cycle in $S$ contributes
$q_C-1=2(b_C-2)+3$ to the sum in~\eqref{eq:component-deficiency}.
By alternation, a nontrivial $\ell$-edge path has $\ell-1\le2b_C$ internal vertices; its available-vertex set contributes one more if $C\in S$. Isolates contribute zero. Therefore $i+\sum_{C\in S}(q_C-1) \le 2\sum_{C\text{ a path}}b_C +2\sum_{C\text{ a cycle}}(b_C-2)+3|S| =2s+3|S|\le2s+3h_*$.
Here the cycle terms added for $C\notin S$ are nonnegative. Apply~\eqref{eq:component-deficiency}.
\end{proof}

\subsection{An auxiliary graph for component matching}
To avoid storing all component adjacencies, keep the available vertices separate. For arbitrary $Q$, order them as $p_{C,1},\ldots,p_{C,q_C}$, introduce $q_C-1$ private vertices outside $V$, and add the internal path
\begin{equation}
 p_{C,1}-z_{C,1}-p_{C,2}-z_{C,2}-\cdots-z_{C,q_C-1}-p_{C,q_C}.
 \label{eq:component-path}
\end{equation}
For $q_C=1$, this is a single vertex. The auxiliary graph $L_Q$ contains these paths and every eligible original connector; unavailable original vertices remain as isolates. Define $b_0=\sum_C(q_C-1)$. Then $L_Q$ is simple, with $n+b_0\le2n$ vertices and at most $m+2b_0$ edges. It need not be bipartite. Also $b_0\le f$: a cycle contributes its length minus one to both quantities, a nontrivial path contributes one to $b_0$ and its length to $f$, and an isolate contributes zero.

\begin{lemma}
\label{lem:auxiliary-matching}
For every simple $2$-matching $Q$,
\begin{equation}
 \nu(L_Q)=b_0+h_*.
 \label{eq:auxiliary-value}
\end{equation}
Every matching $M$ in $L_Q$ yields original connectors representing a matching $J$ of $K_Q$, with $|J|\ge|M|-b_0$.
\end{lemma}
\begin{proof}
\emph{Lifting a component matching.} Choose original witnesses for a maximum matching of $K_Q$. In each component at most one available vertex, say $p_{C,j}$, is used externally. Match $z_{C,i}$ to $p_{C,i}$ for $i<j$, and to $p_{C,i+1}$ for $i\ge j$. This matches all private vertices while leaving $p_{C,j}$ free. Choose any $j$ in an unused component. The resulting matching has $b_0+h_*$ edges.

\emph{Extracting connectors.} For an arbitrary auxiliary matching $M$, let $X$ be its external edges and $t_C$ their number of endpoints in component $C$. Each internal matched edge uses a private vertex and an available vertex not occupied by $X$. Its component therefore contributes at most $\min\{q_C-1,q_C-t_C\}=(q_C-1)-(t_C-1)_+$
internal edges, where $x_+=\max\{x,0\}$. With $\Phi=\sum_C(t_C-1)_+$, we obtain $|M|\le b_0+|X|-\Phi$.
While a component meets at least two edges of $X$, delete one of them. Each deletion reduces the current excess by at least one, so at most the initial $\Phi$ edges are deleted. The remaining connectors form $J$ with $|J|\ge|X|-\Phi\ge|M|-b_0$. Since $|J|\le h_*$, this also proves the upper bound in~\eqref{eq:auxiliary-value}.
\end{proof}

If $|M|\ge\gamma\nu(L_Q)$ for $0<\gamma\le1$, extraction and cycle opening give
\begin{equation}
 |F|=f+|J|\ge\gamma(f+h_*)+(1-\gamma)(f-b_0)\ge\gamma(f+h_*).
 \label{eq:whole-forest-transfer}
\end{equation}
Since $f\ge b_0$, the extra term is nonnegative. This proves the approximation for the whole forest, even though the extracted connectors need not approximate $h_*$ multiplicatively.

\subsection{Three-call algorithm and approximation guarantee}
\begin{algorithm}[H]
\caption{Unweighted path cover by component matching}
\label{alg:optimized-joining}
\begin{algorithmic}[1]
\Require Unweighted edge stream of $G=(V,E)$; matching ratios $0<\alpha,\beta,\gamma\le1$.
\Ensure A path cover $(V,F)$.
\State Compute a matching $A$ in $G$ with $|A|\ge\alpha\nu(G)$.
\State Compute a matching $B$ in $G-A$ with $|B|\ge\beta\nu(G-A)$.
\State $Q\gets A\cup B$; fix its component labels and available-vertex sets.
\State Prepare the auxiliary-path records defining the virtual stream of $L_Q$.
\State Compute a matching $M$ in the virtual stream of $L_Q$ with $|M|\ge\gamma\nu(L_Q)$.
\State $J\gets\{uv\in M:u,v\in V\}$, the original connector edges of $M$.
\While{some fixed component $C$ of $(V,Q)$ meets at least two edges of $J$}
  \State Delete any one of those edges from $J$.
\EndWhile
\State $F\gets Q$.
\For{each cycle component $C$ of $(V,Q)$}
  \If{$C$ meets its unique connector of $J$ at vertex $v$}
    \State Delete from $F$ either edge of $C$ incident with $v$.
  \Else
    \State Delete from $F$ any edge of $C$.
  \EndIf
\EndFor
\State $F\gets F\cup J$; \Return $(V,F)$.
\end{algorithmic}
\end{algorithm}

\begin{theorem}
\label{thm:five-sevenths}
For $0<\alpha,\beta,\gamma\le1$, Algorithm~\ref{alg:optimized-joining} returns a path cover $F$ satisfying
\[
 |F|\ge\gamma\frac{1+2\alpha+2\beta}{7}\rho(G).
\]
When all three calls have error $0\le\delta<1$, the ratio is at least $\frac{(1-\delta)(5-4\delta)}7=\frac57-\frac97\delta+\frac47\delta^2$.
\end{theorem}
\begin{proof}
Put $Q=A\cup B$ and $s=|B|-2k$. Splitting an optimum path cover gives $2|A|\ge\alpha\rho(G)$. Applying \Cref{lem:compensation} with unit weights gives $2|B|\ge\beta(\rho(G)-|A|+k)$, and hence
\[
 |A|+2|B|-k\ge\beta\rho(G)+(1-\beta)(|A|-k)\ge\beta\rho(G),
\]
because $|A|\ge2k$. Thus
\begin{equation}
 3f=2|A|+(|A|+2|B|-k)+(|B|-2k)\ge(\alpha+\beta)\rho(G)+s.
 \label{eq:optimized-compensation}
\end{equation}
Substitute $s\le3f-(\alpha+\beta)\rho(G)$ into~\eqref{eq:optimized-blocking}: $(1+2\alpha+2\beta)\rho(G)\le7f+5h_*\le7(f+h_*)$.
Equation~\eqref{eq:whole-forest-transfer} now proves the first claim. Substitute $\alpha=\beta=\gamma=1-\delta$ for the second.
\end{proof}
Section~\ref{sec:matching-limits} contrasts this auxiliary call with a third matching in the residual graph.

\section{Semi-streaming implementations and guarantees}
\label{sec:streaming}
For weighted algorithms, assume rational weights with at most $L_{\rm in}$ bits per numerator and denominator. A word has $O(L_{\rm in}+\log(n/\eps))$ bits. At fixed accuracy and $L_{\rm in}=O(\polylog n)$, the bounds below give semi-streaming space.

Let $T_u(N,\delta),S_u(N,\delta)$ be monotone bounds on passes and words for a deterministic error-$\delta$ cardinality-matching routine on at most $N$ vertices. Let $T_w(\delta),S_w(n,\delta,W)$ denote the corresponding bounds for a deterministic routine on simple graphs with integer weights in $[1,W]$.

\begin{corollary}
\label{cor:joining-streaming}
For $0<\eps<1/3$, set $\delta=7\eps/9$. Algorithm~\ref{alg:optimized-joining}, with $\alpha=\beta=\gamma=1-\delta$, gives a $(5/7-\eps)$-approximation in $2T_u(n,\delta)+T_u(2n,\delta)$ passes and $S_u(2n,\delta)+O(n)$ words.
An optional cleanup scan makes the forest maximal without weakening the bound.
\end{corollary}
\begin{proof}
Store $A$ as a matched-partner array to filter the second call. Since $Q=A\cup B$ has at most $n$ edges, its components, available vertices, and auxiliary-path records fit in $O(n)$ words.

A virtual scan of $L_Q$ first generates its $2b_0$ internal edges, then scans $G$ for eligible connectors. This costs one original scan and $O(n)$ extra work. Connectors keep their original endpoint pairs, so no component-pair deduplication is needed. Component computation, extraction, cycle opening, and reporting the forest require no further scan.

Substituting $\delta=7\eps/9$ in \Cref{thm:five-sevenths} gives ratio at least $5/7-\eps+28\eps^2/81$. For cleanup, add an edge when both endpoint degrees are at most one and their current components differ. Degrees and connectivity only increase, so a rejected edge never becomes feasible. One scan, using degrees and union--find in $O(n)$ words, therefore makes the forest maximal.
\end{proof}
The deterministic matching algorithm in~\cite{MitrovicEtAl2025} improves~\cite{FischerMitrovicUitto2022} and gives an $O(\eps^{-6})$-pass implementation using $O(n\eps^{-6})$ words. For its $(1+\kappa)$-approximation convention, choose $\kappa=\delta/(1-\delta)$ so that $1/(1+\kappa)=1-\delta$. The preprint~\cite{IzumiKitamuraYamaguchi2025} gives $O(\eps^{-4})$ passes and $O(n\log n)$ bits. 

\begin{corollary}
\label{cor:tsp12}
For an $n$-vertex $(1,2)$-TSP instance with $n\ge3$ and $0<\eps<1/3$, the algorithm returns a tour of cost at most $(9/7+\eps+5/(7n))\OPT_{12}$, with the pass and space bounds of \Cref{cor:joining-streaming}.
\end{corollary}
\begin{proof}
Apply the algorithm to $G_1$. Its forest has at least $\lambda\rho(G_1)$ edges for $\lambda=5/7-\eps$. By \Cref{lem:tsp-conversion}, completion costs at most $(2-\lambda)\OPT_{12}+\lambda$. Use $\OPT_{12}\ge n$ and $\lambda\le5/7$. Completing the stored paths needs no new scan; each connector costs at most two.
\end{proof}

\begin{theorem}
\label{thm:streaming}
For $0<\eps<1/3$, weighted path cover and Max-TSP admit deterministic $(2/3-\eps)$-approximations in $\poly(1/\eps)$ passes independent of $n$ and $O\bigl(n\log(n/\eps)\poly(1/\eps)\bigr)$ words.
For Max-TSP with $n<4/\eps$, the algorithm stores the rounded graph and solves a degree-two relaxation offline; otherwise it uses two matching calls.
\end{theorem}
\begin{proof}
\emph{Matching primitive.} The weighted matching algorithm in~\cite{HuangSu2023} gives $T_w(\delta)=\poly(1/\delta)$ and $S_w(n,\delta,W)=O(n\log(2W)\poly(1/\delta))$ words. Its streaming implementation of the weight reduction in~\cite{GuptaPeng2013} keeps weight-range dependence in memory, not passes. We write $\log(2W)$ to include $W=1$.

\emph{Normalization and path covers.} Choose reciprocal powers of two $\theta,\delta\in(\eps/8,\eps/4]$. Find the maximum edge weight $w_{\max}$ (zero if $E=\varnothing$) in one pass; if it is zero, any feasible output suffices. Otherwise set $\tau=\theta w_{\max}/n$, $\widehat w(e)=\lfloor w(e)/\tau\rfloor$, and $W=\lfloor n/\theta\rfloor$.
Positive rounded weights lie in $[1,W]$. A cover or tour has at most $n$ edges, so rounding loses at most $\theta w_{\max}$. If $\widehat\rho$ is the optimum rounded cover weight, then
$\tau\widehat\rho\ge(1-\theta)\rho_w(G)$, since $w_{\max}\le\rho_w(G)$.
Run Algorithm~\ref{alg:deletion} with weights $\widehat w$ and $\alpha=\beta=1-\delta$, ignoring zero rounded edges in its matching calls. By \Cref{thm:deletion}, $w(F)\ge\tau\widehat w(F) \ge\tfrac23(1-\delta)\tau\widehat\rho \ge\tfrac23(1-\delta)(1-\theta)\rho_w(G)$.
This is at least $(2/3-\eps)\rho_w(G)$.

\emph{Tours.} Removing a lightest edge of an optimum tour gives
$\rho_w(G)\ge(1-1/n)\OPT_{\mathrm{MaxTSP}}$. For $n\ge4/\eps$, completing the returned paths retains a fraction at least $\tfrac23(1-\delta)(1-\theta)(1-1/n)\ge2/3-\eps$, because $\delta+\theta+1/n\le3\eps/4$.
For $n<4/\eps$, compute a maximum-weight simple $2$-matching in the stored rounded graph. A rounded optimum tour is feasible for this relaxation. Deleting a lightest edge per cycle retains at least two thirds of the relaxation value~\cite{FisherNemhauserWolsey1979,MoranNewmanWolfstahl1990}. Complete the paths. Since $w_{\max}\le\OPT_{\mathrm{MaxTSP}}$, the original output weight is at least $\frac23(1-\theta)\OPT_{\mathrm{MaxTSP}}$.

\emph{Resources.} Normalize on the fly and filter $A$ during the second call. Storing $A\cup B$ allows cycle breaking in memory. This uses $2T_w(\delta)+O(1)$ passes and $S_w(n,\delta,W)+O(n)$ words. Reporting connector weights requires one further scan; reporting the tour does not. In the small branch, $n=O(1/\eps)$ and $\widehat w(e)=O(\eps^{-2})$, so the stored graph and polynomial auxiliary workspace fit $n\poly(1/\eps)$ words.
\end{proof}

\noindent\textbf{Other pass tradeoffs.}
The deterministic bound in \Cref{thm:streaming} gives a constant pass count for each fixed target below $2/3$. Its accuracy dependence can be large: the comparison in~\cite{IzumiKitamuraYamaguchi2025} reports exponent $39$ for the matching primitive. The randomized routine in~\cite{Assadi2025} instead gives $2/3-\eps$ in $O(\log n/\eps)$ passes with probability $1-n^{-\Omega(1)}$. Fresh randomness in the second call gives its guarantee conditional on $A$; a union bound controls both calls.

Lower unweighted targets need less work. The greedy component method in Appendix~\ref{sec:greedy-joining} gives $7/10-\eps$ using two matching calls and one scan. A matching followed by greedy forest extension gives $2/3-\eps$ in $T_u(n,3\eps)+1$ passes (Appendix~\ref{app:seeded}).

Appendix~\ref{app:unit-capacity-transfer} obtains $3/4-\eps$ unweighted and $4/5-\eps$ on triangle-free inputs via~\cite{AhnGuha2014}. The pass bound is $O(\eps^{-4}\log^2(1/\eps)\log n)$. The unit-weight support bound justifies rounding in memory; the weighted comparison in Appendix~\ref{app:bmatching} remains conditional on support availability. No determinism claim is added for these imported algorithms.

Given an explicit starting $2$-matching $Q$, Appendix~\ref{app:conditional-component-query} estimates the optimized joining value $f+h_*$ to additive error $\eta n$ in subquadratic time. Constructing $Q$ within that budget is not part of this result.

\section{Algorithmic limits and open questions}
\label{sec:limits-summary}
\subsection{Limits of residual matching and contraction}
\label{sec:matching-limits}
We consider exact residual matching and repeated contraction. These limitations concern the stated rules, not all semi-streaming algorithms.

\noindent\textbf{Three residual matchings.}
Start with $H_0=\varnothing$. For $i=1,2,3$, compute a maximum-weight matching $M_i$ in $G-H_{i-1}$ and set $H_i=H_{i-1}\cup M_i$. The ratio $\rho_w(H_3)/\rho_w(G)$ has infimum $2/3$ over inputs and exact matching choices, even after optimal recombination of retained edges.

For a unit-weight example, take vertices $0,\ldots,6$ and edge set $M_1\cup M_2\cup M_3\cup R$, where
\[
 M_1=\{12,34,56\},\quad M_2=\{15,26\},\quad
 M_3=\{16,25\},\quad R=\{01,23\}.
\]
These choices are maximum: $M_1$ matches six of seven vertices; after its deletion every edge meets $\{1,2\}$, so each later matching has at most two edges. The retained graph $H_3$ consists of the complete graph on $\{1,2,5,6\}$, the edge $34$, and the isolated vertex $0$. Thus $\rho(H_3)=4$, whereas the Hamiltonian path $0-1-5-6-2-3-4$ gives $\rho(G)=6$. A weight perturbation makes all three choices unique while the ratio tends to $2/3$ (\Cref{thm:three-round}). Our structured third call escapes this restriction by representing the missing connectors $01$ and $23$ in its auxiliary graph.

\noindent\textbf{Repeated contraction.}
The iterative rule in~\cite{AlipourFarokhnejadMoemke2025} retains selected edges. Before each new call, it deletes unused edges incident with internal path vertices, contracts the paths, and discards loops.

Its ratio infima are $2/3$ unweighted and $5/8$ weighted (Appendix~\ref{sec:contraction}). Every exact execution attains at least the corresponding fraction, and examples approach it. The unweighted construction uses particular maximum-matching choices. The weighted construction makes both initial choices unique; further rounds add no positive weight, and every tour containing the selected paths has ratio tending to $5/8$. Thus further rounds do not improve the worst-case guarantees, answering the iteration question for exact calls.

\subsection{Open questions}
\label{sec:improvement-directions}
\noindent\textbf{Sublinear-time estimation.}
The quadratic-query lower bound in~\cite{ChenKannanKhanna2020} and \Cref{lem:tsp-conversion} rule out a scheme giving $(1-\eta)\rho(G)-\eta n\le s\le\rho(G)$ in truly subquadratic time for every fixed $\eta>0$. Appendix~\ref{app:hardness} gives the corresponding computational restriction. Neither bound determines the best retention coefficient or running-time exponent.

\Cref{thm:order-concentration} covers correlated uniform orders. Reversing an order merely transposes $B_\pi$, but nonuniform input-dependent orders could behave differently: an order along a known path retains all its edges. Can an algorithm find useful orders without learning that path, or obtain a better constant from another implicit representation?

\noindent\textbf{Semi-streaming construction.}
The compensation proof uses vertex-disjoint alternating cycles and the degree bound left after deleting one matching. Replacing that matching by $A\cup B$ need not preserve these properties. Better starting relaxations or compatible multiple attachments may improve the unweighted $5/7$ or weighted $2/3$ guarantee. The $4/5$ approach in Appendix~\ref{app:component-extensions} would require a suitable streaming triangle-free $2$-matching primitive on general graphs. The deletion rule's completed-tour infimum remains open.

\noindent\textbf{Use of AI tools.}
Generative AI tools assisted with exploring examples, proof ideas, and tight instances; testing the algorithms on small graphs; editing the exposition; and checking references and parameters. The authors remain responsible for all mathematics, citations, and submitted text.

\label{last-main-page}
\clearpage
\bibliographystyle{plainurl}
\bibliography{references}
\label{last-reference-page}
\clearpage
\appendix
\crefalias{section}{appendix}
\crefname{appendix}{Appendix}{Appendices}
\Crefname{appendix}{Appendix}{Appendices}
\section{Residual matching and deletion-rule tightness}
\label{sec:limits}
Use the residual sequence $H_i$ from Section~\ref{sec:matching-limits}. The following theorem allows arbitrary postprocessing within $H_3$.

\begin{theorem}
\label{thm:three-round}
The infimum of $\rho_w(H_3)/\rho_w(G)$ over positive-optimum inputs and legal exact choices is $2/3$. Ratios approaching $2/3$ occur with unique optima in all three calls. Hence no postprocessing restricted to the edges retained by three exact calls has a larger uniform path-cover guarantee.
\end{theorem}
\begin{proof}
The first two matchings contain a forest of weight at least $2\rho_w(G)/3$ by \Cref{thm:deletion}, giving the lower bound. For the upper bound, use the nine-edge graph and sets $M_1,M_2,M_3,R$ displayed in Section~\ref{sec:matching-limits}. Give their edges weights $1+3t$, $1+2t$, $1+t$, and $1$, respectively, where $0\le t<1/3$.
The displayed matching in each round has maximum possible cardinality and uses only the heaviest available edges. For $t>0$, these edges are strictly heavier than all others, making every choice unique.

The retained graph has path-cover size four, as shown in Section~\ref{sec:matching-limits}. Its four heaviest edges sum to $4+11t$, attained by the path $2-1-5-6$ together with $34$.

In $G$, any six-edge forest is a Hamiltonian path from $0$ to $4$, with last segment $2-3-4$. The only two possibilities pass through $1-5-6-2$ or $1-6-5-2$, of total weights $6+10t$ and $6+8t$. Any forest with at most five edges weighs at most $5+13t<6+10t$. Hence $\frac{\rho_w(H_3)}{\rho_w(G)}=\frac{4+11t}{6+10t}\longrightarrow\frac23$. At $t=0$, equality is attained.
\end{proof}
\begin{proposition}
\label{prop:deletion-tight}
The deletion rule has path-cover ratio infimum $2/3$ for every exact implementation, even when both matching calls have unique optima. It differs from the degree-two-relaxation algorithm in~\cite{MoranNewmanWolfstahl1990}.
\end{proposition}
\begin{proof}
The lower bound is \Cref{thm:deletion}. The family in \Cref{thm:three-round} gives the upper bound even for an optimum forest in a larger retained graph, hence also for the two-call output.
To distinguish the algorithms, take unit-weight edges $E=\{01,12,15,23,26,34,56\}$. The maximum choices $A=\{12,34,56\}$ and $B=\{15,26\}$ give a four-cycle and edge $34$, hence four edges after breaking. The unique maximum simple $2$-matching is the six-edge path $0-1-5-6-2-3-4$: retaining six edges must lower both degree-three vertices $1,2$, forcing deletion of their common edge $12$. The relaxation algorithm therefore returns six edges, not four.
\end{proof}
A fourth residual call collects $R$. For unit weights, \Cref{fig:joining} shows how joining and cleanup instead recover a Hamiltonian path using these missing edges.

\section{Repeated contraction: guarantees and tightness}
\label{sec:contraction}
For a matching $M$, contract each pair $uv$ to a supervertex $[uv]$, leave unmatched vertices unchanged, discard loops, and retain parallel edge copies with their original identities and weights. Denote the quotient by $G/M$. Lifting a selected copy recovers its original edge. An \emph{ambient matching} may pair nonadjacent vertices without adding an edge. If $\pi_M:V\to V(G/M)$ is the quotient map, then $|\pi_M^{-1}(S)|\le2|S|$ for every quotient-vertex set $S$.

Write $F_i$ for the lifted union after $i$ rounds of the iterative rule in Section~\ref{sec:matching-limits}. Since selected edges are never removed, we call the iteration \emph{frozen}. Every call in a legal exact execution is optimal for its cardinality or weight objective. This is the iterative rule of~\cite{AlipourFarokhnejadMoemke2025} with approximate calls replaced by exact ones; quotient edges are identified with their original copies.

\subsection{The unweighted two-round guarantee}
\label{app:unweighted-two-round}
We reprove the two-round guarantee of~\cite{AlipourFarokhnejadMoemke2025} using Tutte--Berge. The contraction bound below also allows ambient matchings.

\begin{lemma}
\label{lem:degree-bounded-matching}
Let $H$ be a loopless multigraph of maximum degree at most four, with degrees counting edge copies. If $S=\{v\in V(H):\deg_H(v)\ge3\}$, then $|E(H)|\le3\nu(H)+|S|$.
\end{lemma}
\begin{proof}
By Tutte--Berge~\cite{QuWest2023}, choose $U\subseteq V(H)$ with $2\nu(H)=|V(H)|+|U|-o(H-U)$, where $o$ counts odd-order components. At most $3|U|+|S\cap U|$ edges meet $U$, since vertices outside $S$ have degree at most two and those in $S$ at most four. For a component $C$ of $H-U$, summing degrees gives $|E(C)|\le|V(C)|+|S\cap V(C)|$. For even $|V(C)|$, this is at most $3|V(C)|/2+|S\cap V(C)|$; for odd $|V(C)|\ge3$, it is at most $3(|V(C)|-1)/2+|S\cap V(C)|$. A singleton has no internal edge. Summing yields $|E(H)|\le3|U|+\frac32(|V(H)|-|U|-o(H-U))+|S|=3\nu(H)+|S|$.
\end{proof}

\begin{lemma}
\label{lem:contraction-matching-bound}
For every simple graph $G$ and ambient matching $M$ on $V(G)$, $\nu(G/M)\ge(\rho(G)-|M|)/3$.
\end{lemma}
\begin{proof}
Fix a maximum path cover $P$. Let $H$ be the spanning subgraph of $G/M$ consisting of the surviving copies of $P$. It has $\rho(G)-|M\cap P|$ edges and maximum degree at most four. An unmatched vertex, or the supervertex of an edge in $M\cap P$, has degree at most two. Hence only pairs in $M\setminus P$ can produce degree at least three. By \Cref{lem:degree-bounded-matching}, $3\nu(G/M)\ge3\nu(H)\ge\rho(G)-|M\cap P|-|M\setminus P|=\rho(G)-|M|$.
\end{proof}

\begin{corollary}
\label{cor:two-round-unweighted}
Let $0\le\delta<1$. If $M_1$ and $M_2$ are error-$\delta$ cardinality matchings in $G$ and $G/M_1$, respectively, their lifted union is a path cover with at least $\frac{(1-\delta)(4+\delta)}6\rho(G)\ge\frac23(1-\delta)\rho(G)$ edges.
\end{corollary}
\begin{proof}
A lifted second-round edge joins at most two edges of $M_1$, and each contracted vertex meets at most one second-round edge. Thus the output consists of paths of length at most three. Since $|M_1|\ge(1-\delta)\rho(G)/2$, \Cref{lem:contraction-matching-bound} gives $|M_1|+|M_2|\ge\frac{1-\delta}{3}\rho(G)+\frac{2+\delta}{3}|M_1|\ge\frac{(1-\delta)(4+\delta)}6\rho(G)$.
\end{proof}

\subsection{The unweighted iteration barrier}
\label{sec:iteration}
\begin{theorem}
\label{thm:iteration-barrier}
The unweighted frozen-contraction ratio has infimum $2/3$ over instances and legal exact executions. For every integer $r\ge2$, a graph $G_r$ has a legal execution of ratio $4r/(6r-1)$ with no edge surviving after round two.
\end{theorem}
\begin{proof}
Corollary~\ref{cor:two-round-unweighted} gives at least $2\rho(G)/3$ edges after two rounds; later rounds only add edges.

\emph{Construction and matching choices.} For the upper bound, let $V(G_r)=\{a_i,b_i,c_i,d_i,x_i,y_i:1\le i\le r\}$ and
\[
\begin{split}
 E(G_r)={}&\{a_ib_i,b_ic_i,c_id_i,x_iy_i,x_ib_i,d_ia_i:1\le i\le r\}\\
 &\cup\{b_iy_{i-1},a_ic_{i-1}:2\le i\le r\}\cup\{b_1c_r\}.
\end{split}
\]
Choose $M_1=\{a_ib_i,c_id_i,x_iy_i:1\le i\le r\}$ and, after contraction, the copies of $M_2=\{b_ic_i:1\le i\le r\}$. The first matching is perfect. Every quotient edge meets one of the $r$ vertices $[a_ib_i]$, so the $r$-edge matching $M_2$ is maximum.

\emph{Termination.} The lifted output consists of $a_i-b_i-c_i-d_i$ and $x_i-y_i$, with $4r$ edges. Its degree-two vertices are $b_i,c_i$. Every unused edge except $d_ia_i$ meets one of these vertices and is deleted. Each $d_ia_i$ becomes a loop when its path is contracted. Thus no third-round edge survives.

\emph{Optimum.} The sequence
\[
\begin{aligned}
 &y_r,x_r,b_r,y_{r-1},x_{r-1},b_{r-1},\ldots,y_1,x_1,b_1,\\
 &c_r,d_r,a_r,c_{r-1},d_{r-1},a_{r-1},\ldots,c_1,d_1,a_1
\end{aligned}
\]
is a Hamiltonian path: every vertex occurs once and consecutive vertices are adjacent. Therefore $\rho(G_r)=6r-1$ and $4r/(6r-1)\to2/3$.
\end{proof}
The construction uses the displayed choices among maximum matchings; it does not rule out a better secondary tie-breaking rule.

\subsection{A fractional certificate for weighted path covers}
\label{sec:forest}
\label{app:contraction}
\noindent\textbf{Matching certificates.}
For a loopless multigraph $H$, a vector $z\in\mathbb R^{E(H)}$ has one coordinate per edge copy, and $z(A)=\sum_{e\in A}z_e$. For $S\subseteq V(H)$, let $H[S]$ be the induced subgraph and set $E_H(S)=E(H[S])$. Let $\delta_H(v)$ be the set of incident edge copies at $v$. Edmonds' matching polytope~\cite{Edmonds1965} is
\[
\begin{split}
\MP(H)=\{z\ge0:\;&z(\delta_H(v))\le1\quad(v\in V(H)),\\
 &z(E_H(S))\le(|S|-1)/2\quad(S\subseteq V(H),\ |S|\ge3\text{ odd})\}.
\end{split}
\]
This is the convex hull of matching incidence vectors, so
\begin{equation}
 \nu_w(H)=\max_{z\in\MP(H)}\sum_{e\in E(H)}w(e)z_e.
 \label{eq:matching-cert}
\end{equation}
If $H$ is bipartite, nonnegativity and the degree inequalities suffice~\cite{Schrijver2003}: summing over the smaller side of $H[S]$ gives each odd-set inequality.

\begin{theorem}
\label{thm:forest-contraction}
For every nonnegatively weighted linear forest $P$ and ambient matching $M$ on $V(P)$,
\[
 \nu_w(P)+\nu_w(P/M)\ge\frac58w(P).
\]
The constant is best possible.
\end{theorem}
\begin{proof}
We construct fractional matchings $x$ in $P$ and $y$ in $P/M$ with $x_e+y_e=5/8$ on every original edge. Here $y_e$ refers to its surviving copy, or is zero when the edge is contracted. Equation~\eqref{eq:matching-cert} then proves the bound.

Put $D=E(P)\cap M$. Call a surviving edge \emph{special} if it shares an endpoint with an edge of $D$, and \emph{ordinary} otherwise. A \emph{$D$-supervertex} is the contraction of an edge of $D$. Set
\[
 x_e=\begin{cases}5/8,&e\in D,\\3/8,&e\text{ special},\\1/2,&e\text{ ordinary},\end{cases}
 \qquad
 y_e=\begin{cases}1/4,&e\text{ special},\\1/8,&e\text{ ordinary}.\end{cases}
\]
\emph{Before contraction.} At a vertex incident with $D$, the $x$-sum is at most $5/8+3/8=1$; elsewhere it is at most $1/2+1/2$. Since $P$ is bipartite, $x\in\MP(P)$.

\emph{After contraction.} Put $H=P/M$. An uncontracted vertex or $D$-supervertex has at most two incident edges, of $y$-sum at most $1/2$. Any other supervertex has at most four, each of coefficient at most $1/4$. Thus the degree inequalities hold.

For an odd set $S\subseteq V(H)$, put $s=|S|\ge3$ and $U=\pi_M^{-1}(S)$. Let $d_S$ count its $D$-supervertices and put $q=|E_H(S)|$. The forest $P[U]$ contains exactly these $q$ surviving edges and the $d_S$ contracted edges. Hence $q+d_S\le |U|-1\le2s-1$. Let $r_S$ count the special edges in $E_H(S)$. Each meets a $D$-supervertex, with at most two per such vertex, so $r_S\le2d_S$. Therefore
\[
 8y(E_H(S))=2r_S+(q-r_S)=q+r_S
 \le2s-1+d_S\le3s-1\le4s-4.
\]
Thus $y(E_H(S))\le(s-1)/2$ and $y\in\MP(H)$. The combined coefficient is $5/8$ on every edge, completing the certificate.

For equality, take $P=v_0\cdots v_7$ with consecutive weights $1,1,1,2,1,1,1$ and $M=\{v_0v_1,v_2v_5,v_3v_4,v_6v_7\}$. Then $w(P)=8$, $\nu_w(P)=4$, and $\nu_w(P/M)=1$: all surviving edges have weight one and meet $[v_2v_5]$. Here $v_2v_5$ is an ambient pair, not an edge of $P$.
\end{proof}
\subsection{Weighted iteration and tour completion}
\label{sec:weighted-barrier}
\begin{theorem}
\label{thm:weighted-barrier}
Every exact weighted implementation of the two-round contraction rule or its frozen iteration has path-cover ratio infimum $5/8$. There is a family of complete weighted graphs with unique first and second maximum-weight matchings on which every path-preserving tour completion has ratio tending to $5/8$, even after further frozen rounds.
\end{theorem}
\begin{proof}
For an optimum path cover $P^*\subseteq G$, the matching optima in $G$ and $G/M_1$ are at least those in $P^*$ and $P^*/M_1$. The forest inequality gives the $5/8$ lower bound, and later rounds cannot decrease the selected weight.

Fix $0<t<1/3$. Let $P=v_0\cdots v_7$ have the weights from \Cref{thm:forest-contraction}. Add vertices $a,a',b,b'$ and give the edges $v_2v_5,v_4a,aa',v_6b,bb'$ weight $t$. Keep the path weights and give all other pairs weight zero. Write $K_t$ for this complete graph.

\emph{First round.} The unique maximum-weight matching is
\[
 M_1=\{v_0v_1,v_2v_5,v_3v_4,v_6v_7,aa',bb'\},\qquad w(M_1)=4+3t.
\]
A matching uses at most three $t$-edges: they form an isolated edge and two two-edge paths. Its path contribution is at most four, or at most three if it uses $v_4a$ or $v_6b$, since deleting vertex $v_4$ or $v_6$ from $P$ leaves maximum matching weight three. Equality requires $v_2v_5,aa',bb'$ and the unique remaining path matching $\{v_0v_1,v_3v_4,v_6v_7\}$. This matching is perfect, so zero-weight edges introduce no further choices.

\emph{Second round.} Every surviving unit-weight edge meets $[v_2v_5]$, so a quotient matching contains at most one. Without such an edge it weighs at most $2t$. Only $v_1v_2$ leaves both attachment edges $v_4a,v_6b$ available. Hence the unique quotient optimum lifts to
\[
 M_2=\{v_1v_2,v_4a,v_6b\},\qquad w(M_2)=1+2t.
\]
It is perfect in the quotient. Its union with $M_1$ consists of
\[
 v_0-v_1-v_2-v_5,\qquad v_3-v_4-a-a',\qquad v_7-v_6-b-b'
\]
and weighs $5+5t$. The only unused positive edges are $v_2v_3,v_4v_5,v_5v_6$. Each meets an internal path vertex, so none survives the next deletion step or belongs to a tour containing these paths. Further frozen rounds add no positive weight, and every path-preserving completion weighs $5+5t$.

\emph{Optimum.} The path $P$ together with $aa'$ and $bb'$ has weight $8+2t$ and admits a zero-weight completion. Any subgraph of maximum degree two that uses $v_2v_5$, $v_4a$, or $v_6b$ must omit an edge of $P$, losing at least one. Its weight is at most $7+5t<8+2t$. Without those three edges, its weight is at most $8+2t$. Therefore
\[
 \rho_w(K_t)=\OPT_{\mathrm{MaxTSP}}(K_t)=8+2t,\qquad
 \frac{5+5t}{8+2t}\longrightarrow\frac58.
\]
On the positive-edge subgraph, the rule stops after round two with the same path-cover optimum.
\end{proof}

The universal lower bound is for path covers; for tours, the theorem supplies a tightness obstruction for path-preserving completion.

\subsection{Approximate calls and \texorpdfstring{$(1,2)$}{(1,2)}-TSP}
\label{app:contraction-variants}
For error-$\delta$ calls, the two-round lower bounds are $\frac23(1-\delta)$ unweighted by \Cref{cor:two-round-unweighted} and $\frac58(1-\delta)$ weighted by \Cref{thm:forest-contraction}. Exact optima also satisfy these contracts, so the contracts alone cannot exclude the bad executions. A particular implementation may impose additional selection rules.

For the unweighted graph $G_r$, give its edges cost one and every other pair cost two. The output has $4r$ cost-one edges in $2r$ paths. No cost-one edge joins endpoints of distinct output components, so every path-preserving completion costs $8r$. Its displayed Hamiltonian path gives a tour of cost at most $6r+1$; vertex $y_r$ has only one cost-one neighbor, so this is optimal. Thus the ratio approaches $4/3$ for this legal execution.

\section{Supporting comparisons and extensions}
\label{app:comparisons}
\subsection{Conditional comparison via capacitated \texorpdfstring{$b$}{b}-matching}
\label{app:bmatching}
A capacitated $b$-matching permits incident multiplicity at most $b_v$ at vertex $v$ and multiplicity at most $c_e$ on edge $e$. Thus $b_v=2,c_e=1$ gives integral simple $2$-matchings. Let $Z$ be the optimum integral weight and put $R(n,\eta)=O(\eta^{-4}\log^2(1/\eta)\log n)$.

The fractional algorithm in~\cite{AhnGuha2014} makes $R(n,\eta)$ passes. Its additional workspace is $O(\min\{m,\sum_v b_v\}\poly(1/\eta,\log n))$. It requires polynomially bounded integers and $3/\sqrt n<\eta\le1/16$. Its feasible output $y$ has weight $(1-O(\eta))Z$ and support $E_+=\{e:y_e>0\}$ satisfying
\begin{equation}
 \sum_{e\in E_+}w(e)c_e\le16R(n,\eta)Z.
 \label{eq:unit-support-source}
\end{equation}
The integral rounding uses this support, takes $O(|E_+|/\eta^2)$ space, and returns weight at least $(1-\eta)\sum_e w(e)y_e-\eta Z$. The following weighted comparison assumes that the support and its values can be stored.

\begin{proposition}
\label{prop:bmatching}
If the fractional output and edge weights are retained in $n\poly(1/\eta,\log n)$ words with $|E_+|$ within the same bound, its rounding and cycle breaking give $2/3-\eps$ for weighted path cover and Max-TSP in $O(\eps^{-4}\log^2(1/\eps)\log n)$ passes and $n\poly(1/\eps,\log n)$ words.
\end{proposition}
\begin{proof}
Normalize with error $\theta=\Theta(\eps)$ as in \Cref{thm:streaming}, and apply the assumed support interface to the normalized instance. The source's integral rounding uses only the stored support, giving value $(1-O(\eta))Z$ without another input scan. A path cover or tour is feasible for the degree-two relaxation. Deleting a lightest edge from each cycle retains at least two thirds of its weight; complete the paths for Max-TSP. Choose $\eta=\Theta(\eps)$ sufficiently small. If $n\le9/\eta^2$, store the rounded graph and solve the relaxation exactly; otherwise the source's size condition holds.
\end{proof}
The workspace bound alone does not imply the support hypothesis. Appendix~\ref{app:unit-capacity-transfer} establishes it for unit weights using~\eqref{eq:unit-support-source}.

\subsection{Compatibility of two attachments}
\label{app:attachments}
Removing $uv$ from a cycle exposes exactly $u,v$. Two distinct attachments are therefore compatible with one cut if and only if they are adjacent. Attaching leaves to opposite vertices of a four-cycle requires two cuts.

More generally, let $H=A\cup B$ have $k$ cycle components. Suppose connectors $J$ give a linear forest on its components, with one quotient edge per connector. Every attachment to a nontrivial path must be at an endpoint; if it is used twice, the two endpoints must be distinct. For a cycle used twice, require attachments at distinct adjacent vertices. A cycle used once may attach at any vertex; an isolate may have two incident connectors. Then opening cycles compatibly and adding $J$ gives a linear forest of size $|A|+|B|-k+|J|$. Indeed, all original-vertex degrees are at most two, and a cycle in the lifted graph would induce a cycle in the component forest. An unconstrained degree-two relaxation on components omits these endpoint conditions.

\subsection{One matching call and greedy extension}
\label{app:seeded}
Let $A$ be a matching in $G$ and $F\supseteq A$ a maximal path cover. Let $I$ contain its internal vertices, and let $q$ count nontrivial components, so $|F|=|I|+q$. At most $2|I|$ optimum edges meet $I$. By maximality, every other optimum edge joins the two endpoints of one component of $F$, giving at most $q$ more. Thus $\rho(G)\le2|I|+q=|F|+|I|$.
On each nontrivial $\ell$-edge component, at most $\lceil\ell/2\rceil$ edges belong to $A$, and $\ell-1\le2(\ell-\lceil\ell/2\rceil)$. Isolates contribute zero, so $|I|\le2(|F|-|A|)$ and $\rho(G)\le3|F|-2|A|$. Also $|I|\le|F|$, giving $\rho(G)\le2|F|$. Therefore
\[
 |F|\ge\max\{\rho(G)/2,\ (\rho(G)+2|A|)/3\}.
\]
An error-$\delta$ matching gives $|F|\ge(2-\delta)\rho(G)/3$. One greedy extension scan constructs such an $F$, in $T_u(n,3\eps)+1$ passes at target $2/3-\eps$. Greedy weighted path extension itself is classical~\cite{MoranNewmanWolfstahl1990}.
After cycle opening, the forest need not contain $A$. The seed bound therefore does not by itself improve the joining guarantees.

\subsection{Hardness of near-optimal value estimation}
\label{app:hardness}
Suppose that for every fixed $\eta>0$ a polynomial-time algorithm returns a number $s$ with $(1-\eta)\rho(G)-\eta n\le s\le\rho(G)$. By \Cref{lem:tsp-conversion}, $\OPT_{12}\le2n-s\le(1+2\eta)\OPT_{12}+1$. Given target $\eps$, choose $\eta=\eps/4$ and handle $n<2/\eps$ by exhaustive search; this is constant-size work at fixed accuracy. The resulting TSP value estimator has factor $1+\eps$. The gap form of APX-hardness of $(1,2)$-TSP~\cite{PapadimitriouYannakakis1993} rules out these deterministic estimators unless $P=NP$, or bounded-error randomized estimators unless $NP\subseteq BPP$. The same holds for explicit covers. This computational restriction is distinct from the unconditional pair-query lower bound in~\cite{ChenKannanKhanna2020}.

\subsection{Failure of compensation for a \texorpdfstring{$2$}{2}-matching}
The compensation lemma uses vertex-disjoint alternating cycles of two matchings. If its second argument is replaced by an arbitrary simple $2$-matching, even a triangle-free one, the unchanged inequality is false. Take $A=\varnothing$, $B=E(C_5)$, unit weights, and a four-edge path $P\subset C_5$. The natural cycle-breaking loss is one, yet $2\nu(G-A)=4<|P|-|A|+1=5$. The original lemma allows arbitrary $A$, so this disproves the proposed extension.

\section{Component matching: cycle lengths and access bounds}
\label{app:component-extensions}
Let $Q$ be an explicit simple $2$-matching in an unweighted graph, with notation as in Section~\ref{sec:optimized-joining}. In particular, $k$ counts cycles, $f=|Q|-k$, and $h_*=\nu(K_Q)$.

\subsection{Cycle lengths and the starting relaxation}
For each integer $g\ge3$, define
\[
 \begin{aligned}
 d_g(Q)&=\sum_{C\text{ a cycle of }Q}(g-|V(C)|)_+,\\
 s_g(Q)&=\sum_{C\text{ a path of }Q}|E(C)|
       +\sum_{C\text{ a cycle of }Q}(|V(C)|-g)_+.
 \end{aligned}
\]
Isolates contribute zero. Thus $|Q|=gk+s_g(Q)-d_g(Q)$.
\begin{theorem}
\label{thm:component-defect}
For any simple $2$-matching $Q$ and integer $g\ge3$,
\begin{equation}
 (g+1)(f+h_*)\ge(g-1)|Q|+\rho(G)-d_g(Q).
 \label{eq:component-defect}
\end{equation}
For $0<\gamma\le1$, a ratio-$\gamma$ matching in $L_Q$, projected and opened as in Section~\ref{sec:optimized-joining}, therefore returns
\[
 |F|\ge\frac{\gamma}{g+1}
 \bigl((g-1)|Q|+\rho(G)-d_g(Q)\bigr).
\]
In particular, if $0\le\theta<1$, $|Q|\ge(1-\theta)\rho(G)$, and $Q$ has no cycle shorter than $g$, the returned ratio is at least $\gamma\frac{g-(g-1)\theta}{g+1}$.
\end{theorem}
\begin{proof}
Use the set $S$ from \Cref{lem:component-deficiency}. A cycle of length $\ell$ contributes $q_C-1=\ell-1\le(g-1)+(\ell-g)_+$. A path contributes at most one to the sum over $S$, and its internal-vertex count is at most its number of edges. An isolate contributes zero. Therefore
\[
 i+\sum_{C\in S}(q_C-1)\le s_g(Q)+(g-1)|S|
 \le s_g(Q)+(g-1)h_*.
\]
Equation~\eqref{eq:component-deficiency} now gives $\rho(G)\le f+s_g(Q)+(g+1)h_*$. Rearranging and using $f=|Q|-k$ and $|Q|=gk+s_g(Q)-d_g(Q)$ gives
\[
 (g+1)(f+h_*)\ge gf+\rho(G)-s_g(Q)
 =(g-1)|Q|+\rho(G)-d_g(Q).
\]
Apply \eqref{eq:whole-forest-transfer} to obtain the approximate-call form. The last statement has $d_g(Q)=0$.
\end{proof}
Offline, a maximum triangle-free simple $2$-matching has $|Q|\ge\rho(G)$, so $g=4$ gives $4/5$. That relaxation is polynomial-time solvable~\cite{Paluch2025}, but the resulting ratio is below the general $6/7$ guarantee in~\cite{BermanKarpinski2006}. If the input has girth at least $g$, an ordinary maximum simple $2$-matching gives $g/(g+1)$. Here girth is the shortest cycle length, infinite for a forest; the theorem itself requires the cycle-length condition only on $Q$.

\subsection{Unit-weight streaming bounds via integral rounding}
\label{app:unit-capacity-transfer}
Choose $0<\eta\le1/16$ and use the fractional algorithm and rounding from Appendix~\ref{app:bmatching} with $b_v=2$ and $c_e=w(e)=1$. Integral solutions are simple $2$-matchings, with optimum $Z\le n$. Since $w(e)c_e=1$ for every edge, \eqref{eq:unit-support-source} gives $|E_+|=\sum_{e\in E_+}w(e)c_e\le16R(n,\eta)Z\le16R(n,\eta)n$. Both the edges and fractional values therefore fit in $n\poly(1/\eta,\log n)$ words. For general weights and capacities, the source controls a weighted sum rather than the support size; here no additional assumption is needed.

Rounding uses only the stored support and returns $Q$ of size $(1-O(\eta))Z\ge(1-O(\eta))\rho(G)$. Its $O(|E_+|/\eta^2)$ workspace remains within the stated bound. Apply the same source algorithms to $L_Q$, now with $b_v=c_e=w(e)=1$, to obtain a matching of ratio $1-O(\eta)$. Since $L_Q$ has at most $2n$ vertices, the same support argument applies; each virtual pass costs one original scan by \Cref{cor:joining-streaming}.

All weights and capacities are one or two, so the source's requirement of polynomially bounded integers holds. For $n>9/\eta^2$, both calls satisfy the source's size condition, because $L_Q$ retains the original vertices. Otherwise store $G$ and solve both matching problems offline. This branch uses $n\poly(1/\eta)$ words, including polynomial auxiliary workspace.

\begin{corollary}
\label{cor:three-quarters}
For $0<\eps<1/3$, the guarantees in~\cite{AhnGuha2014} and component joining give an unweighted $(3/4-\eps)$-approximation in $O\bigl(\eps^{-4}\log^2(1/\eps)\log n\bigr)$ passes and $n\poly(1/\eps,\log n)$ words.
On triangle-free input graphs the ratio is $4/5-\eps$. More generally, for each fixed $g\ge3$, input girth at least $g$ gives $g/(g+1)-\eps$ in the same resource class. On arbitrary $(1,2)$-TSP instances the general $3/4-\eps$ bound yields factor $5/4+\eps+3/(4n)$.
\end{corollary}
\begin{proof}
The two source calls give $\theta=O(\eta)$ and $\gamma=1-O(\eta)$ in \Cref{thm:component-defect}, hence ratio at least $g/(g+1)-O(\eta)$ whenever $Q$ has no cycle shorter than $g$. An ordinary simple $2$-matching has no cycles of length below three and inherits every girth promise of the input. Set $g=3$ for general graphs, $g=4$ for triangle-free inputs, and choose a sufficiently small constant multiple $\eta=\Theta(\eps)$ for the respective guarantee. The small-input branch uses exact calls. The tour factor follows from \Cref{lem:tsp-conversion} with $\lambda=3/4-\eps$ and $\OPT_{12}\ge n$.
\end{proof}

\subsection{Subquadratic estimation with a supplied \texorpdfstring{$2$}{2}-matching}
\label{app:conditional-component-query}
Store the component labels and positions in~\eqref{eq:component-path}. An original pair requires one query to $G$ only when both vertices are available in different components. Original--private adjacency is determined by path positions; private--private pairs are absent. Thus an auxiliary pair query uses at most one original query. Traversing the supplied $Q$ builds these records in $O(n)$ time and space.

\begin{theorem}
\label{thm:conditional-component-query}
Suppose an explicit simple $2$-matching $Q\subseteq G$ is supplied together with pair-query access to $G$. For every fixed $0<\eta<1$, an algorithm returns with high probability a scalar $v$ such that
\[
 f+h_*-\eta n\le v\le f+h_*\le\rho(G)
\]
in $n^{2-\Omega_\eta(1)}$ time and queries. In particular, for every integer $g\ge3$,
\[
 v\ge\frac{(g-1)|Q|+\rho(G)-d_g(Q)}{g+1}-\eta n.
\]
It need not recover a forest.
\end{theorem}
\begin{proof}
Use \Cref{thm:bks-interface} on $L_Q$, at accuracy $\eta/2$ on at most $2n$ vertices, to obtain $z$ with $\nu(L_Q)-\eta n\le z\le\nu(L_Q)$. Return $v=f+\max\{0,z-b_0\}$. On the success event, \eqref{eq:auxiliary-value} gives $h_*-\eta n\le\max\{0,z-b_0\}\le h_*$. This includes $z<b_0$: then $h_*<\eta n$, and returning the baseline is sufficient. Maximum component joining gives an actual forest of value $f+h_*$, proving the upper bound by $\rho(G)$. Apply \eqref{eq:component-defect} for the last bound. Preprocessing and the constant-overhead oracle preserve the subquadratic runtime.
\end{proof}
The scalar estimator does not construct $Q$ or its component labels. The per-vertex matching oracle in~\cite{BhattacharyaKissSaranurak2023} has a different cost and does not remove the assumption that $Q$ is supplied.

\section{Greedy component joining and its limits}
\label{sec:greedy-joining}
A greedy joining scan gives $7/10-\eps$ with two matching calls. We use the available vertices, eligible connectors, and cycle-opening rule from Section~\ref{sec:optimized-joining}.

Compute matchings $A$ in $G$ and $B$ in $G-A$. Put $H=(V,A\cup B)$, and let $k$ count its cycle components. As in Section~\ref{sec:optimized-joining}, $f=|A|+|B|-k$ denotes the baseline forest size after cycle opening. Set $s=|B|-2k\ge0$.

\noindent\textbf{The joining scan.}
Mark every component of $H$ unused. Scan the original edges once, accepting an eligible connector when both components are unused, then mark them used. Let $J$ be the accepted connectors and $h=|J|$. Labels and availability refer to the fixed graph $H$. The component matching is maximal: a connector between two finally unused components would have been accepted.

Open cycles and add $J$ as in Section~\ref{sec:optimized-joining}, obtaining a valid linear forest $F^+$ of size $f+h$. The optional cleanup scan of \Cref{cor:joining-streaming} uses current component labels and yields a maximal forest $F^\star\supseteq F^+$.

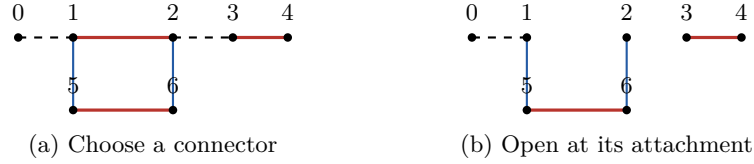
\begin{figure}[hb]
\centering
\begin{tikzpicture}[x=0.66cm,y=0.64cm,every node/.style={font=\small},
 Ae/.style={mone,very thick},Be/.style={mtwo,thick},con/.style={black,thick,dashed}]
\begin{scope}
\coordinate(v0)at(-1.1,1);\coordinate(v1)at(0,1);\coordinate(v2)at(2,1);
\coordinate(v5)at(0,-0.5);\coordinate(v6)at(2,-0.5);\coordinate(v3)at(3.2,1);\coordinate(v4)at(4.3,1);
\draw[Ae](v1)--(v2);\draw[Ae](v5)--(v6);\draw[Ae](v3)--(v4);
\draw[Be](v1)--(v5);\draw[Be](v2)--(v6);
\draw[con](v0)--(v1);\draw[con](v2)--(v3);
\foreach \i in {0,1,2,3,4,5,6}{\fill(v\i)circle(1.6pt);\node[above=3pt]at(v\i){$\i$};}
\node at(1.6,-1.25){(a) Choose a connector};
\end{scope}
\begin{scope}[xshift=6.0cm]
\coordinate(v0)at(-1.1,1);\coordinate(v1)at(0,1);\coordinate(v2)at(2,1);
\coordinate(v5)at(0,-0.5);\coordinate(v6)at(2,-0.5);\coordinate(v3)at(3.2,1);\coordinate(v4)at(4.3,1);
\draw[Ae](v5)--(v6);\draw[Ae](v3)--(v4);
\draw[Be](v1)--(v5);\draw[Be](v2)--(v6);\draw[con](v0)--(v1);
\foreach \i in {0,1,2,3,4,5,6}{\fill(v\i)circle(1.6pt);\node[above=3pt]at(v\i){$\i$};}
\node at(1.6,-1.25){(b) Open at its attachment};
\end{scope}
\end{tikzpicture}
\caption{Joining before cycle opening. Solid edges form $A\cup B$; dashed edges in (a) are eligible connectors. Accepting $01$ and removing $12$ gives the five-edge forest in (b). A subsequent cleanup scan can add $23$.}
\label{fig:joining}
\end{figure}

\subsection{Approximation and pass bounds}
Let $\mathcal T$ denote the set of used components; since each connector in $J$ pairs two distinct components, $|\mathcal T|=2h$. Let $t_C$ and $t_P$ count the cycle and path components in $\mathcal T$, with isolated vertices counted under $t_P$, so that $t_C+t_P=2h$. For each component $C$, let $b_C=|B\cap E(C)|$. Let $z$ count the used isolated vertices, and define the surplus of unused cycles as $u=\sum_{C\text{ unused cycle}}(b_C-2)\ge0$.

\Needspace{5\baselineskip}
\begin{lemma}
\label{lem:blocking}
For any edge-disjoint matchings $A,B$ and any maximal matching of their eligible component graph,
\[
 \rho(G)\le f+2s+5t_C+3t_P-z-2u
          =f+2s+10h-2t_P-z-2u.
\]
\end{lemma}
\begin{proof}
Let $N_{\mathcal T}=\sum_{C\in\mathcal T}|V(C)|$. Let $I$ contain the internal vertices of unused paths and put $i=|I|$. Delete the vertices in used components and in $I$. The survivors lie in unused cycles, endpoint pairs of unused nontrivial paths, or unused isolates. No input edge joins surviving vertices from different original components of $H$, since it would be an eligible connector between two unused components.

Fix an optimum path cover $P$, and let $q$ count unused nontrivial paths. At most $2N_{\mathcal T}+2i$ edges of $P$ meet deleted vertices. Within an unused cycle $C$, at most $|V(C)|-1$ remain, even if $G$ contains chords; at most one remains between the endpoints of each unused path. Hence
\[
 |P|\le2N_{\mathcal T}+2i+q+
 \sum_{\substack{C\text{ unused}\\C\text{ a cycle}}}(|V(C)|-1).
\]
After cycle opening, used components contribute $N_{\mathcal T}-|\mathcal T|$ edges; unused paths contribute $i+q$. Thus
\[
 f=N_{\mathcal T}-|\mathcal T|+i+q+
 \sum_{\substack{C\text{ unused}\\C\text{ a cycle}}}(|V(C)|-1),
\]
and $\rho(G)\le f+N_{\mathcal T}+|\mathcal T|+i$.

A used cycle contributes $|V(C)|+1=2(b_C-2)+5$ to $N_{\mathcal T}+|\mathcal T|$. By alternation, a nontrivial $\ell$-edge path has $\ell\le2b_C+1$: if used, it contributes at most $2b_C+3$; if unused, its internal vertices contribute at most $2b_C$ to $i$. A used isolate contributes two, not the three charged by $3t_P$, giving the correction $-z$. Unused cycles and unused isolates contribute nothing. Consequently,
\[
\begin{aligned}
 N_{\mathcal T}+|\mathcal T|+i
 &\le2\sum_{C\text{ a path}}b_C
 +2\sum_{\substack{C\text{ used}\\C\text{ a cycle}}}(b_C-2)+5t_C+3t_P-z\\
 &=2s+5t_C+3t_P-z-2u.
\end{aligned}
\]
The term $-2u$ removes the unused-cycle surplus included in $2s$. Substitute in the preceding bound and use $t_C+t_P=2h$.
\end{proof}

\begin{theorem}
\label{thm:joining}
If $|A|\ge\alpha\nu(G)$ and $|B|\ge\beta\nu(G-A)$, where $0<\alpha,\beta\le1$, then
\[
 |F^\star|\ge |F^+|\ge
 \frac{1+3\alpha+3\beta}{10}\rho(G)+\frac{t_P}{5}+\frac{z}{10}+\frac{u}{5}.
\]
In particular, exact matching calls give a $7/10$ approximation, and two error-$\delta$ calls give $7/10-3\delta/5$.
\end{theorem}
\begin{proof}
Equation~\eqref{eq:optimized-compensation} gives $3f\ge(\alpha+\beta)\rho(G)+s$. Substitute $s\le3f-(\alpha+\beta)\rho(G)$ into \Cref{lem:blocking} and use $|F^+|=f+h$ to obtain $(1+2\alpha+2\beta)\rho(G)+3f\le10|F^+|-2t_P-z-2u$. Since $s\ge0$, also $3f\ge(\alpha+\beta)\rho(G)$. Rearranging gives the bound; cleanup only adds edges. Set $\alpha=\beta=1$ or $\alpha=\beta=1-\delta$ for the stated special cases.
\end{proof}

Appendix~\ref{app:attachments} gives conditions for two compatible attachments. In the upper family of Appendix~\ref{app:joining-limit}, $t_P=z=u=0$, so the corrections above do not improve the $7/10$ bound.

\noindent\textbf{Pass tradeoff.}
Setting $\delta=5\eps/3$ for $0<\eps<1/3$ yields a $(7/10-\eps)$-approximation. The forest $F^+$ uses $2T_u(n,\delta)+1$ streaming passes and $S_u(n,\delta)+O(n)$ words of memory; the cleanup scan that gives $F^\star$ adds one final pass. Connected components are determined from the stored matching edges.

\subsection{A worst-case upper bound of \texorpdfstring{$17/24$}{17/24}}
\label{app:joining-limit}
The greedy method has legal exact choices and scan orders with ratios tending to $17/24$. This uses matching ties and concerns only the greedy method, not \Cref{thm:five-sevenths}.

Fix an integer $h\ge1$ and put $\ell=4h-1$. For $0\le i<\ell$, construct a four-cycle on distinct vertices $p_i,q_i,x_i,y_i$, with $A_i=\{p_iq_i,x_iy_i\}$ and $B_i=\{p_ix_i,q_iy_i\}$.
Create $\ell$ additional disjoint edges $a_jb_j$ ($0\le j<\ell$) and one isolated vertex $z$. Put
\[
 A=\bigcup_{i=0}^{\ell-1}A_i\ \cup\ \{a_jb_j:0\le j<\ell\},
 \qquad B=\bigcup_{i=0}^{\ell-1}B_i.
\]
The first $2h$ cycles are designated as distinguished. The graph has $n=6\ell+1$ vertices, and every input edge has unit weight.

\noindent\textbf{Additional edges and a Hamiltonian path.}
List the $4h$ vertices $p_i,q_i$ of the distinguished cycles as $d_{2i}=p_i$, $d_{2i+1}=q_i$ for $0\le i<2h$. Let $Z_0=(z)$ and $Z_{j+1}=(a_j,b_j)$ for $0\le j<\ell$. Begin the vertex sequence
\[
 Z_0,d_0,Z_1,d_1,\ldots,Z_{4h-1},d_{4h-1},x_0,y_0,
\]
expanding each block $Z_j$ in its displayed order. For $j=0,\ldots,2h-2$, append
\[
 p_{2h+j},x_{2h+j},y_{2h+j},q_{2h+j},x_{j+1},y_{j+1}.
\]
Call the complete sequence $\Pi$. It visits every vertex of $G$ exactly once. Add all consecutive edges of $\Pi$ to $A\cup B$. Finally, add the edges $p_{2j+1}x_{2j}$ for $0\le j<h$, and place these $h$ special connectors first in the streaming order of the joining scan. There are no other edges in $G$.

\noindent\textbf{Matching optimality.}
The matching $A$ contains $3\ell=\lfloor n/2\rfloor$ edges, and is therefore maximum. Let $S=\{p_i,q_i:0\le i<\ell\}$. Every edge outside $A$ meets $S$: this holds for $B$, the added path edges of $\Pi$, and the special connectors. Thus $S$ is a vertex cover of $G-A$ of size $2\ell$, which is attained by the matching $B$. Hence, both matching calls are exact.

\noindent\textbf{Output and optimum.}
The components of $A\cup B$ are exactly the $\ell$ four-cycles, $\ell$ isolated edges, and the isolate $z$. Every intercomponent edge touches a distinguished cycle. The first $h$ scanned connectors are accepted and mark all $2h$ distinguished cycles used. No subsequent connector in the stream joins two unused components, so $|J|=h$. Therefore, $|F^+|=|A|+|B|-\ell+h=4\ell+h=17h-4$. The Hamiltonian path $\Pi$ gives $\rho(G)=n-1=6\ell=24h-6$, and hence
\[
 \frac{|F^+|}{\rho(G)}=\frac{17h-4}{24h-6}\longrightarrow\frac{17}{24}.
\]
Together with \Cref{thm:joining}, this places the greedy rule's worst-case ratio infimum in $[7/10,17/24]$.

\subsubsection{Maximal outputs in the upper-bound family}
\label{app:cleanup-upper}
In the construction above, configure the cycle openings as follows:
remove $x_i y_i$ for every even distinguished index $i<2h$;
remove $p_i x_i$ for every odd distinguished index $i<2h$;
and remove $x_i y_i$ for every undistinguished index $i\ge2h$.
Each opening at a used cycle is incident with its attachment: $x_{2j}$ for even indices and $p_{2j+1}$ for odd indices.

After adding the selected connectors, every $p_i,q_i$ has degree two. For even distinguished and undistinguished cycles, deleting $x_iy_i$ leaves both vertices internal. For odd distinguished cycles, deleting $p_ix_i$ leaves $q_i$ internal, while the connector supplies the second edge at $p_i$. Thus every vertex of the vertex cover $S=\{p_i,q_i:0\le i<\ell\}$ of $G-A$ is saturated.

Every edge of $\Pi$ outside $A\cup B$ meets $S$, so none can be added. In used cycles, each removed edge meets the saturated attachment vertex. In unused cycles, the removed $x_iy_i$ joins endpoints of the same path, so adding it would create a cycle. No other edges exist. Hence $F^+$ is maximal and $|F^\star|=|F^+|=17h-4$ after cleanup. These openings preserve the limit $17/24$; other openings may allow more edges.

\section{An explicit exponent for matching-value estimation}
\label{app:bks-parameters}
We give an explicit exponent for the estimator in~\cite{BhattacharyaKissSaranurak2023}. Source locators refer to the full version arXiv:2302.05030v3. The notation below is local.

Fix the additive accuracy $0<a<1$ of \Cref{thm:bks-interface}. Here $\gamma=a/4$ is the accuracy scale, $2K+1$ bounds the augmenting-path length, $q$ is the survival fraction, $\psi$ is the source's smallness parameter, $\sigma$ is the minimal relative gain per successful call, $T$ bounds the iterations per call, $R$ bounds the total number of calls, and $\xi(a)$ is the resulting base exponent. Define
\begin{align*}
 \gamma&=a/4, & K&=\lceil1/\gamma\rceil, &
 q&=\gamma^2(2K+2)^{-(2K+2)},\\
 \psi&=\frac{q}{4(10^8K+1)}, &
 \sigma&=10^{-8}\psi^{5^{4K+3}}, &
 T&=\left\lceil\frac{3(K+2)^2}{\sigma}\right\rceil,\\
 R&=\left\lceil(K+1)(\sigma^{-1}+1)\right\rceil, &&&
 \xi(a)&=\frac{1}{20\,9^{T(R+1)}}.
\end{align*}
These parameters depend only on $a$. We show that the positive exponent $\xi(a)$ gives time $\widetilde O_a(N^{2-\xi(a)})$ on $N$ vertices.

\noindent\textbf{Layer survival and thresholds.}
The source's Algorithm~4 uses $K=\lceil1/\gamma\rceil$ and calls the augmentation procedure at lengths $2i+1$, $0\le i\le K$, with density parameter $\gamma^2$. If the current matching admits at least $\gamma^2N$ vertex-disjoint augmenting paths of length $2i+1$, fix such a collection before sampling the layers. By their Lemma~5.3, each path survives independently with probability $(2i+2)^{-(2i+2)}$. The expected number of survivors is at least $qN$, so a Chernoff bound gives at least $qN/2$ survivors with high probability for sufficiently large $N$ (depending on $a$). Without the candidate-path hypothesis, the source procedure may return failure. Use the same $\psi$ for all calls. Their equation~(1) in Section~5.1.1 sets
\[
 \psi_j=10^{-8}\psi^{5^{4j+3}},\qquad
 \delta_j=\psi^{5^{4j+1}}.
\]
For $i\le K$, the loss bound in Corollary~5.16 is at most
\[
 \sum_{j=0}^{i-1}\frac{\delta_{j+1}}{\psi_j}N
 \le K\,10^8\psi N.
\]
Since $\delta_0=\psi^5\le\psi$ and
$(K10^8+1)\psi=q/4<q/2$, on the survival event above the number of paths exceeds that loss plus $\delta_0N$. Thus our explicit $\psi$ satisfies the source's sufficient-smallness condition under the candidate-path hypothesis. Whenever an augmentation call succeeds, the paragraph after Corollary~5.17 guarantees an increase of at least $\psi_iN\ge\sigma N$.

\noindent\textbf{Iteration counts.}
Claim~5.13 bounds backtracks at layer $2j$ by $1/\psi_j$. The displayed calculation in Corollary~5.14 consequently bounds all iterations of a length-$2i+1$ call by
\[
 (i+2)\left(\sum_{j=0}^{i-1}\psi_j^{-1}+1\right)+(i+1)
 \le (K+2)(K/\sigma+1)+K+1\le T.
\]
There are at most $\sigma^{-1}$ successful calls, since matching size increases by at least $\sigma N$ and never exceeds $N/2$. Every nonfinal round of Algorithm~4 contains a success, and every round has $K+1$ calls. Hence the total call count is at most $R$.

\noindent\textbf{Exponent bookkeeping.}
Theorem~5.2 and equation~(6) of Section~5.2 multiply the matching-oracle exponent by at most $9^T$ per successful call. Start with exponent $\xi(a)$. Even allowing this multiplication on every call and once more for a final unsuccessful call gives $\xi(a)9^{T(R+1)}=1/20$. Every intermediate exponent is therefore at most $1/20<1/5$, as required by Definition~5.1. Each preprocessing call costs at most $\widetilde O_a(N^{2-\xi(a)})$; their number is an $a$-dependent constant. The final matching oracle has query exponent at most $1+1/20$, so the sampling step of Corollary~6.2 fits within the same bound. That corollary has additive error $4\gamma N=aN$ and the one-sided value guarantee. Exceptional bounded $N$ can be read in full, absorbed into the $a$-dependent constant.

For~\eqref{eq:query-composition-time}, use $a=\eta/4$, or $a=\eps/8$ for TSP. That bound includes confidence amplification of one matching-value estimate and the cost of one vertex order. The displayed $\xi$ is a conservative choice, not an optimized exponent.
\end{document}